\RequirePackage{fix-cm}
\documentclass[smallcondensed]{svjour3}     
\smartqed  
\usepackage{graphicx}
\def\E{{\mathrm{E}}}

\def \R{{{\rm I{\!}\rm R}}}
\def \N{{\mathbb{N}}}
\def \Z{{\mathbb{Z}}}

\def\E{{\mathrm{E}}}

\def \R{{{\rm I{\!}\rm R}}}
\def \N{{\mathbb{N}}}
\def \Z{{\mathbb{Z}}}

\usepackage{amssymb}
\usepackage{enumitem}
\usepackage{graphicx}

\usepackage{natbib}
\usepackage{multicol}
\usepackage{amssymb}

\usepackage{amsfonts,amsthm,amsmath}

\usepackage{graphicx,epstopdf}
\usepackage{setspace}
\usepackage{color}
\usepackage{multirow}
\usepackage{rotating}
\usepackage[ruled,vlined,linesnumbered,noresetcount]{algorithm2e}
\usepackage[noend]{algpseudocode}
\usepackage{color}
\usepackage{xcolor}

\begin{document}

\title{The modified Yule-Walker method for  multidimensional infinite-variance periodic autoregressive model of order 1}


\author{Prashant Giri \and
Aleksandra Grzesiek \and
Wojciech \.Zu\l awi\'nski \and
       S. Sundar \and
       Agnieszka Wy{\l}oma{\'n}ska
}


\institute{Prashant Giri and  S. Sundar \at
            Department of Mathematics\\ Indian Institute of Technology Madras\\
	        600036 Chennai, India \\
             \email{prashantnet2013@gmail.com},
             \email{slnt@iitm.ac.in}\\
           \and
           Aleksandra Grzesiek, Wojciech \.Zu\l awi\'nski (Corresponding Author) and Agnieszka Wy{\l}oma{\'n}ska,   \at
            Faculty of Pure and Applied Mathematics, Hugo Steinhaus Center\\ Wroc{\l}aw University of Science and  Technology\\
           Wybrze\.ze Wyspia{\'n}skiego 27\\ 50-370 Wroc{\l}aw, Poland \\        
	       \email{aleksandra.grzesiek@pwr.edu.pl}
	       \email{ wojciech.zulawinski@pwr.edu.pl\ (Corresponding Author)}
	       \email{agnieszka.wylomanska@pwr.edu.pl}}

\date{}

\maketitle

\begin{abstract}
The time series with periodic behavior, such as the periodic autoregressive (PAR) models belonging to the class of the periodically correlated processes, are present in various real applications. In the literature, such processes were considered in different directions, especially with the Gaussian-distributed noise. However, in most of the applications, the assumption of the finite-variance distribution seems to be too simplified. Thus, one can consider the extensions of the classical PAR model where the non-Gaussian distribution is applied. In particular, the Gaussian distribution can be replaced by the infinite-variance distribution, e.g. by the $\alpha-$stable distribution. In this paper, we focus on the multidimensional $\alpha-$stable PAR time series models. For such models, we propose a new estimation method based on the Yule-Walker equations. However, since for the infinite-variance case the covariance does not exist, thus it is replaced by another measure, namely the covariation. In this paper we propose to apply two estimators of the covariation measure. The first one is based on moment representation (moment-based) while the second one  - on the spectral measure representation (spectral-based).
The validity of the new approaches are verified using the Monte Carlo simulations in different contexts, including the sample size and the index of stability of the noise. Moreover, we compare the moment-based covariation-based method with spectral-based covariation-based technique. Finally, the real data analysis is presented.
	
	\keywords{periodic autoregression \and heavy-tailed distribution \and covariation \and estimation \and Monte Carlo simulation  }
	
    
    \subclass{92C50 \and 62P10}
\end{abstract}

\section{Introduction}
In many real applications, one can observe the periodic behavior of the data. Very often the periodicity is not visible in the raw time series but in its characteristics, like in the sample autocovariance function. In that case, the model that can be useful to the data description belongs to the class of periodically correlated processes (called also cyclostationary processes).
The periodically correlated processes are useful  in various real applications, including mechanical systems \citep{mech1,antoni2004cyclostationary}, hydrology \citep{hyd1,hyd2}, climatology and meteorology \citep{met1,met2}, economics \citep{broszkiewicz2004detecting,econ2}, medicine and biology \citep{medicine1,medicine2} and many others. The idea of periodically correlated processes was initiated in \citep{Guzdenko_1959_Antonio,gladyshev}, and then extensively extended by many authors.  One of the most known members of the periodically correlated discrete-time models is the periodic autoregressive moving average (PARMA) time series, \citep{aaw8,aaw10}. The one-dimensional PARMA models were examined in many statistical papers \citep{aaw7,aaw17,aaw16,aaw15,aaw14,aaw11,aaw12,aaw13}. The PARMA models are considered as the natural extension of the classical autoregressive moving average (ARMA) time series \citep{brockwelldavis}, where instead of the constant-coefficient the periodic parameters are used. It is worth mentioning, the PARMA models can be also treated as the special case of the time-dependent coefficients ARMA time series \citep{aaw29,aaw28,aaw30}. In the classical version, the PARMA models are based on the Gaussian distribution of the noise.

Although the theory of the finite-variance PARMA models is still extended, the assumption of the finite second moment of the given process (mostly Gaussian distributed) is inappropriate for many real applications. Thus, many theoretical models with non-Gaussian distribution were considered in different applications e.g. finance \citep{fin1}, physics \citep{phys}, electricity market \citep{el}, technical diagnostics \citep{zak1,zak3,new1}, geophysical science \citep{geo1,geo2}, and many others. One can also find the research studies related to PARMA models with the infinite-variance distribution of the noise \citep{kruczek2020detect,kruczek_physica,nowicka_wylomanska}. 
However, when we analyze the time series models without the assumption of the finite second moment of the distribution, the classical methods to the parameters' estimation and statistical investigation can not be directly applied. The main problem in such analysis results from the infinite theoretical autocovariance function that is a base for many statistical methods. Thus, dedicated algorithms need to be introduced. In the literature, one can find the algorithms for the estimation of the parameters for the infinite-variance constant-coefficient ARMA \citep{est2,kruczek_physica} as well as PARMA \citep{kruczek_physica} time series. In many of the cases, the noise is described by the $\alpha-$stable distribution that is the most known member of the infinite-variance class of distributions  \citep{Taqqu}. 

Similar to the multidimensional ARMA models (also called vector ARMA, VARMA) known in the literature \citep{aaw25,aaw24,aaw26,aaw27}, one may also consider the multivariate version of the PARMA models. They can be also treated as an example of the multivariate ARMA models with time-varying coefficients \citep{aaw27_1,aaw27_2,aaw27_3}. In the literature, one may find the research papers devoted to the multidimensional PAR (or PARMA) models \citep{aaw20,aaw22,aaw19,aaw21,aaw23}. However, similar as in one-dimensional case, also in the multivariate one, the mentioned above multivariate PARMA models are based mostly on the multidimensional Gaussian distribution. However, when we analyze the real multivariate data, this assumption seems to be too simple to model different phenomena. Thus, similar to the one-dimensional case, one can consider the multidimensional PARMA models with infinite-variance multidimensional distribution which seems to be more proper for real applications. Very often in real examples, two or even more components of the same process are examined at the same time and the dependence between them is crucial to model the data accurately. What more, when the external sources appear, then the non-Gaussian (impulsive) behavior of the components is observed. The perfect example is the energy market, where the obvious periodicity of the data exists, and the components such as energy price and energy demand are related. In a one-dimensional case, the energy data were modeled by using the PARMA time series \citep{broszkiewicz2004detecting} also by using the non-Gaussian $\alpha-$stable distribution \citep{kruczek_physica}. In our recently published paper, we start the research related to the theoretical properties of the multidimensional PAR models with multivariate $\alpha-$stable noise \citep{our_2dim_par}. 

However, when we analyze the PAR models with the assumption of infinite-variance distribution, one can consider the estimation methods that take into consideration the non-Gaussian behavior. In a one-dimensional case the appropriate methods were proposed for stationary AR models \citep{zul,est2,kruczek_physica} and also for PAR time series \citep{kruczek_physica}. In this paper, we extend the research and propose a new estimation method for the multivariate PAR models with infinite-variance multidimensional noise. Here we concentrate on the estimation method which extends the classical Yule-Walker algorithm widely used in the Gaussian case for VAR and PAR time series \citep{brockwelldavis}. However, in the finite-variance case, the Yule-Walker algorithm is based on the autocovariance and cross-covariance function of a given time series. In the infinite-variance case, these functions are infinite, thus, similar to  one-dimensional case \citep{est2,kruczek_physica}, here we replace the classical measure by the alternative one adequate for infinite-variance processes. We use the auto-covariation and cross-covariation functions that are properly defined for $\alpha-$stable models \citep{KokTaq97}. The auto-covariation and cross-covariation functions (as well as other alternative dependency measures) were analyzed in our recent papers in different directions, see for instance \citep{our_2dim_par,sundar1,nasza,nasza2}. The main attention we pay to the PAR(1) models with multidimensional  $\alpha-$stable noise. The validity of the new method is presented for the simulated data and for real time series.

The paper is organized in the following way. In Section \ref{stable} we recall the multidimensional $\alpha$-stable distribution together with the covariation measure used to quantify the dependence within the random vector. In Section 3 we present the definition of a multidimensional periodic autoregressive model based on the $\alpha-$stable distribution, the form of the bounded solution for multidimensional PAR(1) time series together with the formulas for the corresponding auto- and cross-covariation functions. In Section 4 we introduce  new covariation-based estimation procedure for the parameters of multidimensional $\alpha-$stable PAR(1) time series model. Section 5 contains the simulation study. {In Section \ref{realdata} we present the real data analysis.} Section \ref{Conclusions} concludes the paper.

\section{The multidimensional symmetric $\alpha-$stable distribution} \label{stable}

The definition of a general multidimensional $\alpha-$stable random vector $\textbf{Z}=(Z_1,Z_2,\ldots,Z_m)$ can be given using the characteristic function \citep{miller,cova1,Taqqu}, which in a special case of the symmetric vectors ($S\alpha S$), taken under consideration in this paper, takes the following form
\begin{equation} \label{fun_char}
\Phi_{\alpha}(\boldsymbol{\theta})=\exp\left\{-\int_{S_m}|\langle\boldsymbol{\theta},\textbf{s}\rangle|^{\alpha}\Gamma(\mathrm{d}s)\right\}.
\end{equation}
where $0<\alpha\leq2$ is a stability parameter (or stability index), $\Gamma(\cdot)$ is a finite symmetric spectral measure on the unit sphere $S_m$ of ${\R}^m$, and $\langle\cdot,\cdot\rangle$ denotes the scalar product. The measure $\Gamma(\cdot)$ in Eq. (\ref{fun_char}) is called the spectral measure of the random vector $\textbf{Z}$ and uniquely defines the distribution. Moreover, the spectral measure defines the dependence structure between the $\alpha-$stable vector components. For $S\alpha S$ random vector $\textbf{Z}$ we introduce the following notation 
\begin{equation}
    \mathbf{Z} \sim S_\alpha(\Gamma).
\end{equation} 
More properties of the one-dimensional and multidimensional $\alpha-$stable distribution the readers can find for instance in \citep{Taqqu,miller,cova1,zolotarev}.

\subsection{Dependence measure for $\alpha-$stable random variables}

To describe the dependence structure of $S\alpha S$ random vectors one can use the measure called covariation. For a $S\alpha S$ random vector $(X,Y)$ with the stability index $1< \alpha \leqslant 2$, the covariation of $X$ on $Y$ is defined as an intergral with respect to the spectral measure $\Gamma(\cdot)$ of $(X,Y)$, namely it is a real number determined as follows
\begin{equation} \label{CV-def}
	CV(X,Y)=\int_{S_2}s_1s_2^{\langle\alpha-1\rangle}\Gamma(\mathrm{d}s),
\end{equation}
where $x^{\langle a\rangle}$ is a so-called signed power, i.e., $x^{\langle a\rangle}=|x|^a sign(x).$ Alternatively, for all $1\leq q<\alpha$ the covariation can be defined as follows \citep{Taqqu}
\begin{equation}\label{CV}
	CV(X,Y)=\frac{E[XY^{\langle q-1\rangle}]\sigma^{\alpha}_{Y}}{E|Y|^q},
\end{equation}
where $\sigma_{Y}$ is the scale parameter of $Y$, i.e., $\sigma_Y = \left(\int_{S_2}|s_2|^{\alpha}\Gamma(\mathrm{d}s)\right)^{\frac{1}{\alpha}}$. In \citep{Taqqu} the authors proved that covariation defined in Eq. (\ref{CV}) is a value independent on $q$.

The covariation defined above is additive and linear in the first argument, i.e. for the symmetric $\alpha-$stable random variables $X_1, X_2, Y$ and the real numbers $a_1 , a_2$ we have \citep{Taqqu}
\begin{equation*}
CV(a_1 X_1 + a_2 X_2, Y) = a_1 CV(X_1,Y) + a_2 CV(X_2,Y).
\end{equation*}
However, the additivity property in the second argument, i.e. $CV(Y, X_1 + X_2) = CV(Y, X_1) + CV(Y, X_2)$, holds only if $X_1$ and $X_2$ are independent \citep{Taqqu}. Moreover, for the covariation the following scaling property holds \citep{Taqqu}
\begin{equation*}
CV(a_1 X,a_2 Y)= a_1 a_2^{\left \langle \alpha -1 \right \rangle}CV(X,Y).
\end{equation*}
Another important property is the fact that for independent random variables $X$ and $Y$ the covariation reduces to zero, namely $CV(X,Y) = 0$, but the implication in the opposite direction is not true \citep{Taqqu}. Moreover, for $\alpha = 2$ (Gaussian case), the covariation is proportional to the classical second-moment-based covariance measure, i.e. $2 CV(X,Y) = Cov(X,Y)$. It is important to mention that in general the measure is non-symmetric in the arguments, i.e. $CV(X,Y) \neq CV(Y,X)$. Furthermore, for $1<\alpha<2$ the measure given in Eq. (\ref{CV-def}) determines a covariation norm on the linear space of jointly symmetric $\alpha-$stable random variables denoted as follows \citep{Taqqu}
\begin{equation} \label{norm}
\Vert X \Vert_\alpha =  (CV(X,X))^{1/\alpha}.
\end{equation}

The covariation can be also used to measure the interdependence (auto-dependence) within a one-dimensional time series $\{X(t)\}$, $t\in \Z$. In this case, the function is called the auto-covariation and from Eq. (\ref{CV}) it can be expressed in the following form
\begin{equation}\label{auto-CV}
	CV(X(t),X(s))=\frac{E[X(t)X(s)^{\langle q-1\rangle}]\sigma^{\alpha}_{X(s)}}{E|X(s)|^q},
\end{equation}
where $t,s \in \mathbb{Z}$. Moreover, the formula for the auto-covariation can be generalized to the case of the multivariate time series $\{\mathbf{X}(t)\}=\{X_1(t),\ldots,X_m(t)\}$, $t\in \Z$ for which we additionally observe the dependence between the component time series $\{X_i(t)\}$ and $\{X_j(t)\}$ for $i\neq j$. To describe that kind of dependence (cross-dependence) one can use the cross-covariation given, similarly as in the previous case, by the following formula \citep{nasza2,our_2dim_par,nasza_asympt}
\begin{equation}\label{cross-CV}
	CV(X_i(t),X_j(s)=\frac{E[X_i(t)X_j(s)^{\langle q-1\rangle}]\sigma^{\alpha}_{X_j(s)}}{E|X_j(s)|^q},
\end{equation}
where $t,s \in \mathbb{Z}$ and $j \neq i$.

\subsection{Estimation of covariation} 
\subsubsection{Estimation of covariation based on moments} \label{est_cov}
In this subsection, we present how to estimate the auto-covariation and the cross-covariation functions for the stationary $\alpha-$stable processes. In practice, instead of estimating the exact measures given in Eqs. (\ref{auto-CV}-\ref{cross-CV}) one often estimates their normalized versions, i.e. the covariation divided by the scale parameter of the second argument raised to the power of $\alpha$ with $q=1$, see \citep{est2}.

Let $\{x(t)\}$ denote a trajectory of stationary one-dimensional time series, where $t = 1,\ldots,L$ and $L$ is the trajectory length. The estimator of normalized auto-covariation takes the following form
\begin{equation} \label{NCV1}
 \widehat{NCV}(h)  = \widehat{NCV}(X(t),X(t-h)) = \frac{\sum_{t=r}^{l}x(t)sign(x(t-h))}{\sum_{t=r}^{L}\left | x(t) \right |},
 \end{equation}
where $r = max(1,1+h)$ and $l = min(L,L+h)$. For the multidimensional time series   with a trajectory $\{\textbf{x}(t)\}=\{(x_1(t),\ldots,x_m(t))\}$, where $t = 1,\ldots,L$ and $L$ is the trajectory length, we estimate the normalized cross-covariation function using the following formula
\begin{equation} \label{NCV_par}
 \widehat{NCV}_{i,j}(h)  = \widehat{NCV}(X_i(t),X_j(t-h))= \frac{\sum_{t=r}^{l}x_i(t)sign(x_j(t-h))}{\sum_{t=r}^{L}\left | x_j(t) \right |}
 \end{equation}
where $r = max(1,1+h), l = min(L,L+h)$ and $i, j = 1,\ldots,m$. Let us notice that for $i=j$ the estimator given in Eq. (\ref{NCV_par}) simplifies to this presented in Eq. (\ref{NCV1}).

\subsubsection{Estimation of covariation based on spectral measure} \label{est_cov1}

As an alternative to the methodology presented above, in order to estimate the auto-covariation and the cross-covariation, we can directly apply the formula presented in Eq. (\ref{CV-def}). Namely, following the approach proposed by \citep{physica}, the covariation of jointly $S\alpha S$ random vector $(X_1,X_2)$ with spectral measure $\Gamma(\cdot)$ can be estimated as follows
\begin{equation} \label{theo_est_CV}
\widehat{CV}(X_1,X_2) = \sum_{j=1}^{L} s_{1,j}(s_{2,j})^{\langle \alpha-1\rangle} \Gamma^{*}(\textbf{s}_j).
\end{equation} 
In the above-given formula, $\Gamma^{*}(\cdot)$ is a discrete approximation of the spectral measure $\Gamma(\cdot)$ and $\textbf{s}_j = (s_{1,j},s_{2,j})$ denotes the location of spectral measure's masses on unit sphere in $\R^2$. 

It is important to emphasize that an arbitrary spectral measure $\Gamma(\cdot)$ can be approximated by a discrete measure $\Gamma^*(\cdot)$ in such a way that the densities corresponding to both random vectors are uniformly close. For the details of the construction of such approximation, including the number, locations and weights of the point masses, we refer the readers to the results given in \citep{theor_CV_5}.

As one can notice, calculating the covariation using Eq. (\ref{theo_est_CV}) requires the estimation of a discrete spectral measure $\Gamma^{*}$. In the literature, there are variety of methods proposed for this purpose \citep{Cheng_Rachev,PIVATO2003219,OGATA2013248,Mohammad,Sathe}. In the following part of the paper, we decide to apply the widely-used approach introduced by \citep{NOLAN_Panorska_McCulloch}. The chosen estimation procedure is known as the projection method since it is based on one-dimensional projections of the considered random vector. In the two-dimensional case the method was consider also by \citep{NOLAN_Panorska_McCulloch_2}.

{Since the time series considered in our paper is non-stationary, to estimate the auto- or cross-covariation functions of the model we cannot directly apply the estimator mentioned above. However, the estimators used in the procedure introduced in Section \ref{est_method} are based on the ones given above since the proposed method involves estimating the auto- and cross-covariation functions corresponding to a number of stationary sub-samples driven from a trajectory of non-stationary process.}

\section{Periodic autoregressive model with multidimensional $\alpha-$stable noise}
In this section, we consider the periodic autoregressive (PAR) model with $S\alpha S$ distribution. We start our consideration by recalling the general definition of a one-dimensional and multidimensional $\alpha-$stable PAR model of order $p\in{\N}$ (denoted as PAR(p)). Then, we focus on the time series examined in this paper, namely the $m$-dimensional PAR model of order $1$.{
\begin{definition} {\citep{kruczek_physica}}  \label{PAR_def}
 (\textbf{One-dimensional PAR(p) model}) A time series $\{ X(t) \}$, $t\in\Z$ is called the one-dimensional $\alpha-$stable periodic autoregressive model of order $p$ with a period $T\in \N$ if for each $t\in{\Z}$ it satisfies the following equation
	\begin{equation}\label{par_p_1}
	X(t)=\theta_1(t)X(t-1)+\ldots+\theta_p(t)X(t-p)+ Z(t),
	\end{equation}
where $\{Z(t)\}$, $t\in\Z$ is a sequence of independent $S\alpha S$ random variables which characteristic function defined in Eq. (\ref{fun_char}) and $\theta_1(t),\ldots,\theta_p(t)$ are the periodic parameters  (periodicity with respect to  $t\in{\Z}$) with the period equal to $T$.
\end{definition}
\begin{definition} \label{PARm_def}
(\textbf{Multidimensional PAR(p) model})  A time series $\{\textbf{X}(t)\}=\{X_1(t),\ldots,X_m(t)\}$ is called the m-dimensional $\alpha-$stable periodic autoregressive model of order $p$ with a period $T\in \N$ if for each $t\in{\Z}$ it satisfies the following system of equations
	\begin{equation}\label{par_p_m}
	\textbf{X}(t)=\mathbf{\Theta_1}(t)\textbf{X}(t-1)+\ldots+\mathbf{\Theta_p}(t)\textbf{X}(t-p)+\textbf{Z}(t),
	\end{equation}
where $\{\textbf{Z}(t)\} = \{Z_1(t), \ldots, Z_m(t)\}$ is the m-dimensional $S\alpha S$ random vector in ${\R}^m$ with the characteristic function defined in Eq. (\ref{fun_char}) and $\mathbf{\Theta_1}(t),\ldots,\mathbf{\Theta_p}(t)$ are the $m \times m$ coefficient matrices periodic in $t$ with same period equal to $T$. Additionally,  we assume that $\textbf{Z}(t)$ is independent of $\textbf{Z}(t+h)$ for all $h\neq 0$. 
\end{definition}
}

\begin{remark}
If the period $T$ is equal to $1$, then the periodic autoregressive models presented in Definitions \ref{PAR_def} and \ref{PARm_def} reduce to the one-dimensional and $m$-dimensional, respectively, autoregressive (AR) time series with the coefficients constant in time. 
\end{remark}

As we mentioned before, in this paper we focus on multidimensional $\alpha-$stable periodic autoregressive model of order 1. More precisely, we consider the $m$-dimensional time series $\{ \textbf{X}(t)\} = \{X_1(t),\ldots, X_m(t)\}$ satisfying Eq. (\ref{par_p_m}) with $p=1$. For the elements of the coefficient matrix $\mathbf{\Theta}(t)=\mathbf{\Theta}_1(t)$ we propose the following notation
{\begin{equation}
	\Theta(t)= \left[
	\begin{array}{cccc}
	\theta_{11}(t) & \theta_{12}(t) & \ldots & \theta_{1m}(t)\\
	\theta_{21}(t) & \theta_{22}(t) & \ldots & \theta_{2m}(t)\\
	\ddots & \ddots & \ddots & \ddots \\
	\theta_{m1}(t) & \theta_{m2}(t) & \ldots & \theta_{mm}(t)\\	\end{array}
	\right],
	\label{Theta}
\end{equation}}
where ${\theta_{ij}(t)}$ are periodic in $t$ with period equal to $T$ for $i,j \in \{1,\ldots,m\}$. We mention here that the simplified version of the above-defined PAR(1) model, namely the two-dimensional time series with $p=1$, was considered by the authors in \citep{our_2dim_par} and the form of the bounded solution for the $m$-dimensional PAR(1) model is analogous to the one presented in \citep{our_2dim_par} for bivariate case. Namely, based on the theorem formulated in \citep{peiris}, in the Hilbert space of symmetric $\alpha-$stable random variables with $1<\alpha<2$ with the covariation norm defined in Eq. (\ref{norm}), the bounded solution of Eq. (\ref{par_p_m}) with $p=1$ has the form
\begin{eqnarray}\label{main2}
	\mathbf{X}\left(t\right)=\sum_{j=0}^{+\infty}g(t,t-j+1)\mathbf{Z}\left(t-j\right),
\end{eqnarray}
where $g(t,t-j+1)$ is given as follows
\begin{eqnarray}\label{g}
g(t,t-j+1)=
\begin{cases}
I & \mbox{when}~ j=0,\\
\Theta(t)\Theta(t-1)...\Theta(t-j+1) & \mbox{when}~ j>0.
\end{cases}
\end{eqnarray}
Moreover, the conditions guaranteeing the absolute convergence with probability $1$ of the solution for all $t\in \Z$ are given as follows
\begin{eqnarray}\label{cond1}
\sum_{j=0}^{+\infty}|g_{rl}(t,t-j+1)|<+\infty
\end{eqnarray}
for all $r,l\in\{1,2,\ldots,m\}$, where $g_{rl}(t,t-j+1)$ is the $(r,l)$ element of $g(t,t-j+1)$. Below we formulate the expressions for the cross-covariation of the $m$-dimensional $\alpha-$stable PAR(1) model's components. 

\begin{lemma} \label{covariation_general}
		Let $\{\mathbf{X}(t)\}=\{X_1(t),\ldots,X_m(t)\}$, $t\in\Z$ be the bounded solution of Eq. (\ref{par_p_m}) with $p=1$ presented by Eq. (\ref{main2}). Then for $t,s \in \Z$ the cross-covariation is given as follows
	\begin{enumerate}
		\item[(a)] if $s\geq t$
		\begin{multline} \label{CV1}
		\mathrm{CV}(X_r(s),X_l(t))=\\\sum_{j=0}^{+\infty}\int_{S_m}\left(g_{l1}({t},{t-j+1})s_1+g_{l2}({t},{t-j+1})s_2+\ldots+g_{lm}({t},{t-j+1})s_m\right)^{\langle\alpha-1\rangle}\\\left(g_{r1}({s},{t-j+1})s_1+g_{r2}({s},{t-j+1})s_2+\ldots+g_{rm}({s},{t-j+1})s_m\right)\Gamma(ds),
		\end{multline}
		\item[(b)] if $s\leq t$
		\begin{multline}
		\mathrm{CV}(X_r(s),X_l(t))=\\\sum_{j=0}^{+\infty}\int_{S_m}\left(g_{l1}({t},{s-j+1})s_1+g_{l2}({t},{s-j+1})s_2+\ldots+g_{lm}({t},{s-j+1})s_m\right)^{\langle\alpha-1\rangle}\\\left(g_{r1}({s},{s-j+1})s_1+g_{r2}({s},{s-j+1})s_2+\ldots+g_{rm}({s},{s-j+1})s_m\right)\Gamma(ds),
		\end{multline}
	\end{enumerate} 
	where $r,l \in \{1,\ldots,m\}$ and $r \neq l$.
\end{lemma}
\begin{proof} \renewcommand{\qedsymbol}{}
	The proof of Lemma \ref{covariation_general} is given in Appendix A.
\end{proof}
The formulas for the cross-covariation given in Lemma \ref{covariation_general} can be simplified if we assume that the multidimensional $\alpha-$stable PAR(1) model consists of the components dependent only through the noise, i.e. the coefficients matrix given in Eq. (\ref{Theta}) has nonzero elements only on diagonal. In this case, the  components of the  matrix given in Eq. (\ref{g}) are given by
\begin{align} \label{gki1}
    g_{rl}(t,t-j+1)&=0 \quad &\text{for}~r\neq l\\
    g_{rl}(t,t-j+1)&=g_{rr}(t,t-j+1)=\theta_{rr}(t)\theta_{rr}(t-1)\ldots\theta_{rr}(t-j+1) \quad &\text{for}~r= l.
    \label{gki2}
\end{align}
Moreover, since for any $j=NT+k$, $N\in\N$, $k\in\{0,1,..,T-1\}$ one can show that 
\begin{eqnarray}\label{AandD}
	g_{rr}(t,t-j+1)=P_r^Ng_{rr}(t,t-k+1),
\end{eqnarray}
where $P_r=\theta_{rr}(1)\theta_{rr}(2)...\theta_{rr}(T)$, the conditions given in Eq. (\ref{cond1}) simplify to the fact that $|P_r|<1$ for all $r\in\{1,\ldots,m\}$. The formulas for the cross-covariation corresponding to the simplified model are presented in the Lemma given below. 
\begin{lemma} \label{lemma3}
		Let $\{\mathbf{X}(t)\}=\{X_1(t),\ldots,X_m(t)\}$, $t\in\Z$ be the bounded solution of a simplified version of Eq. (\ref{par_p_m}) with $p=1$ given by Eq. (\ref{main2}) with $\theta_{rl}=0$ for $r\neq l$ in Eq. (\ref{Theta}). Then for $t,s \in \Z$ the cross-covariation is given as follows
			\begin{enumerate}
			\item[(a)] if $s\geq t$
			\begin{multline}\label{CV_1}
			\mathrm{CV}(X_r(s),X_l(t))=\\\frac{g_{rr}(s,t+1)\,\int_{S_m}s_l^{\langle \alpha-1\rangle}s_r\Gamma(ds)}{1-P_l^{\langle\alpha-1\rangle}P_r}\sum_{k=0}^{T-1}\left(g_{ll}({t},{t-k+1})\right)^{\langle \alpha-1 \rangle}g_{rr}({t},{t-k+1}),
			\end{multline}
			\item[(b)] if $s\leq t$
			\begin{multline}
			\mathrm{CV}(X_r(s),X_l(t))=\\\frac{g_{ll}(t,s+1)^{\langle \alpha-1\rangle}\,\int_{S_m}s_l^{\langle \alpha-1\rangle}s_r\Gamma(ds)}{1-P_l^{\langle\alpha-1\rangle}P_r}\sum_{k=0}^{T-1}\left(g_{ll}({s},{s-k+1})\right)^{\langle \alpha-1 \rangle}g_{rr}({s},{s-k+1}),
			\end{multline}
		\end{enumerate} 
		where $r,l \in \{1,\ldots,m\}$ and $r \neq l$.
	\end{lemma}
	\begin{proof} \renewcommand{\qedsymbol}{}
	The proof follows directly from Lemma \ref{covariation_general} and the formulas given in Eqs. (\ref{gki1}), (\ref{gki2}) and (\ref{AandD}).
    \end{proof}
Let us notice that since for $r=1,\ldots,m$ the coefficients $\theta_{rr}(t)$ in Eq. (\ref{Theta}) are periodic with respect to $t$ with the period equal to $T$, the cross-covariation given in Lemma \ref{lemma3} is also periodic in $s\in \Z$ for any $s-t\in\Z$ with the period equal to $T$. As it was mentioned, the two-dimensional version of the process considered in this paper was examined in the authors previous article where they analyze the asymptotic relation between the dependence measures in symmetric $\alpha-$stable case, see \citep{our_2dim_par}.  
    
\section{Estimation method for multidimensional PAR(1) model with $\alpha-$stable distribution} \label{est_method}

In this section, we present a new procedure leading to the estimation of the parameters corresponding to the m-dimensional PAR(1) time series with $\alpha-$stable distribution satisfying Eq. (\ref{par_p_m}) with $p=1$. The method is based on covariation (based on moment) defined in Section \ref{est_cov}.

Let $\{\mathbf{X}(t)\}=\{(X_1(t),\ldots,X_m(t)\}$ be the bounded solution of Eq. (\ref{par_p_m}) with $p=1$ given by Eq. (\ref{main2}). Moreover, let us assume that the period is known and equal to $T\in \N$. Now, since each $t \in \Z$ can be expressed as $t=nT+v$, where $v = 1,\ldots,T$ and $n\in\Z$, one can rewrite the corresponding autoregressive equation in a different form
\begin{equation} \label{PAR_eq}
	\textbf{X}(nT+v)=\mathbf{\Theta}(nT+v)\textbf{X}(nT+v-1) + \textbf{Z}(nT+v).
	\end{equation} 
Let us multiply the above equation by vector $\mathrm{sign}(\textbf{X}(nT+v-1))'$, where
\begin{equation*}
\mathrm{sign}(\textbf{X}(nT+v-1))=\left[
\begin{array}{c}
\mathrm{sign}(X_1(nT+v-1))\\
\mathrm{sign}(X_2(nT+v-1))\\
\vdots \\
\mathrm{sign}(X_m(nT+v-1))\\
\end{array}
\right].
\end{equation*}
Then, by taking the expectation we get
\begin{multline} \label{preliminary_eq}
	E[\textbf{X}(nT+v)\mathrm{sign}(\textbf{X}(nT+v-1))']=\\\mathbf{\Theta}(v)E[\textbf{X}(nT+v-1)\mathrm{sign}(\textbf{X}(nT+v-1))'] + E[\textbf{Z}(nT+v)\mathrm{sign}(\textbf{X}(nT+v-1))'],
\end{multline}
 since $\Theta(nT+v)=\Theta(v)$ because $\Theta(t)$ is periodic in $t$ with period $T$. Let us notice that the last component in Eq. (\ref{preliminary_eq}) vanishes since for all $r,l=1,\ldots,m$ we have
 \begin{equation*}
     \E[{Z}_r(nT+v)\mathrm{sign}({X}_l(nT+v-1))] = E[{Z_r}(nT+v)]E[\mathrm{sign}({X}_l(nT+v-1))] = 0.
 \end{equation*}
 Finally, let us multiply Eq. (\ref{preliminary_eq}) by the following $m \times m$  matrix
\begin{equation} \label{matrix}
\left[
\begin{array}{ccccc}
\frac{1}{E|X_1(nT+v-1)|} & 0 & \ldots & 0\\
0 & \frac{1}{E|X_2(nT+v-1)|} & \ldots & 0\\
\ldots & \ldots & \ldots & \ldots\\
0 & 0 & \ldots & \frac{1}{E|X_{m}(nT+v-1)|}\\
\end{array}
\right],
\end{equation}
which leads to the system of equations given below
\begin{equation}
\mathbf{NCV}^v(1)=\mathbf{\Theta}(v)\mathbf{NCV}^{(v-1)}(0), 
\label{system}
\end{equation}
where $\mathbf{NCV}^v(h)$ is the $m \times m$ normalized covariation matrix with the elements taking the following form
\begin{equation} 
 NCV_{r,l}^v(h)  = NCV(X_r(nT+v),X_l(nT+v-h))= \frac{E[X_r(nT+v)sign(X_l(nT+v-h)]}{E \left | X_l(nT+v-h) \right |},
 \end{equation}
where $v = 1,\ldots,T$ and $r,l = 1,\ldots,m $. Let us note that according to \citep{Taqqu} the expected values $\E|X_l(nT+v-h)|$ for $l=1,\ldots,m$ in Eq. (\ref{matrix}) are non-zero.

In our method, to determine the unknown coefficients matrices ${\Theta}(v)$, where $v=1,\ldots,T$, we replace the elements of normalized covariation matrices given in Eq. (\ref{system}) by their estimators, i.e. we consider the following system of equations
\begin{equation}\label{system_est}
\widehat{\mathbf{NCV}}^v(1)= \widehat{\Theta}(v)\widehat{\mathbf{NCV}}^{(v-1)}(0),
\end{equation}
for $v = 1,\ldots,T$. Let us note that since the considered model is non-stationary, to calculate the empirical auto- and cross-covariation we cannot directly apply the estimators given in Section \ref{est_cov}. Therefore, for the matrix ${\mathbf{NCV}^v}(h)$ we need to introduce a modified estimator which takes the following form
\begin{equation} \label{NCV_PARm}
 \widehat{NCV}_{r,l}^v(h)  = \frac{\sum_{n=n_0}^{N-1}x_r(nT+v)sign(x_l(nT+v-h))}{\sum_{n=n_0}^{N-1}\left | x_l(nT+v-h) \right |},
 \end{equation}
where $v = 1,\ldots,T$, $ N = \lfloor \frac{L}{T} \rfloor$ and $L$ is the trajectory length,
\begin{align} \label{r}
n_0 = \left\{\begin{matrix}
0 & \text{\qquad when  $v>0$ and  $v-h>0$,}\\ 
 1 & \text{otherwise,} 
\end{matrix}\right.
\end{align}
$r,l = 1,2,\ldots,m$, and $\{\textbf{x}(t)\} = \{x_1(t),\ldots,x_m(t)\}$, $t=1,\ldots,L$ denotes a sample trajectory of length $L$ corresponding to the m-dimensional time series $\{\mathbf{X}(t)\}=\{(X_1(t),\ldots,X_m(t))\}$. Let us note that this is equivalent to estimating the normalized cross-covariation of sub-samples $\{x_r(v),x_r(T+v),\ldots,x_r((N-1)T+v)\}$ and $\{x_l(v-h),x_l(T+v-h),\ldots,x_l((N-1)T+v-h)\}$ where $r,l=1,\ldots,m$.

When the matrices $\widehat{\mathbf{NCV}}^{(v-1)}(0)$ given in Eq. (\ref{system_est}) are non-singular for $v=1,\ldots,T$, to calculate $\widehat{\mathbf{\Theta}}(v)$ we can use the following formula
\begin{equation}
    \widehat{\Theta}(v)=\widehat{\mathbf{NCV}}^v(1){\widehat{\mathbf{NCV}}^{(v-1)}(0)}^{-1}.
\end{equation}
However, in Appendix B we propose an algorithm which enables solving Eq. (\ref{system_est}) regardless of whether $\widehat{\mathbf{NCV}}^{(v-1)}(0)$ is non-singular or singular but when the system in Eq. (\ref{system_est}) is consistent. In this case, for $v=1,\ldots,T$, we estimate $\widehat{\mathbf{\Theta}}(v)$ by using
\begin{equation*}
\widehat{\Theta}(v)={SOLUTION}_1 \left(  \widehat{\mathbf{NCV}}^{(v-1)}(0), \widehat{\mathbf{NCV}}^v(1)\right), 
\end{equation*}
where $SOLUTION_1$ is an algorithm presented in Appendix B. Moreover, we remind that
\begin{equation*}
\widehat{\mathbf{NCV}}^v(1)=\left[
\begin{array}{cccc}
\frac{\sum_{n=r}^{N-1}x_1(nT+v)sign(x_1(nT+v-1))}{\sum_{n=r}^{N-1}\left | x_1(nT+v-1)) \right |} & \ldots & \ldots & \frac{\sum_{n=r}^{N-1}x_1(nT+v)sign(x_m(nT+v-1))}{\sum_{n=r}^{N-1}\left | x_m(nT+v-1)) \right |}\\[5pt]
\frac{\sum_{n=r}^{N-1}x_2(nT+v)sign(x_1(nT+v-1))}{\sum_{n=r}^{N-1}\left | x_1(nT+v-1)) \right |} & \ldots & \ldots & \frac{\sum_{n=r}^{N-1}x_2(nT+v)sign(x_m(nT+v-1))}{\sum_{n=r}^{N-1}\left | x_m(nT+v-1)) \right |}\\[5pt]
\ldots & \ldots & \ldots & \ldots \\[5pt]
\frac{\sum_{n=r}^{N-1}x_m(nT+v)sign(x_1(nT+v-1))}{\sum_{n=r}^{N-1}\left | x_1(nT+v-1)) \right |} & \ldots & \ldots & \frac{\sum_{n=r}^{N-1}x_m(nT+v)sign(x_m(nT+v-1))}{\sum_{n=r}^{N-1}\left | x_m(nT+v-1)) \right |},
\end{array}
\right]
\label{NCV_1}
\end{equation*}
and
\begin{equation*}
\widehat{\mathbf{NCV}}^{(v-1)}(0) =
\left[ \begin{array}{cccc}
 1 & \ldots & \ldots & \frac{\sum_{n=r}^{N-1}x_1(nT+v-1)sign(x_m(nT+v-1))}{\sum_{n=r}^{N-1}\left | x_m(nT+v-1)) \right |}\\[5pt]
\frac{\sum_{n=r}^{N-1}x_2(nT+v-1)sign(x_1(nT+v-1))}{\sum_{n=r}^{N-1}\left | x_1(nT+v-1)) \right |} & 1 & \ldots & \frac{\sum_{n=r}^{N-1}x_2(nT+v-1)sign(x_m(nT+v-1))}{\sum_{n=r}^{N-1}\left | x_m(nT+v-1)) \right |}\\[5pt]
\ldots & \ldots & \ldots & \ldots \\[5pt]
\frac{\sum_{n=r}^{N-1}x_m(nT+v-1)sign(x_1(nT+v-1))}{\sum_{n=r}^{N-1}\left | x_1(nT+v-1)) \right |} & \ldots & \ldots & 1 
\end{array}
\right],
\label{NCV_0}
\end{equation*}
where $r$ is given in Eq. (\ref{r}), $v = 1,\ldots,T$ and $ \{\textbf{x}(t)\} = \{x_1(t),\ldots,x_m(t)\}$, $t=1,\ldots,L$ is a sample trajectory of length $L$ corresponding to the m-dimensional time series $\{\mathbf{X}(t)\}=\{(X_1(t),\ldots,X_m(t))\}$.

Moreover, we introduce a new estimation technique based on Yule-Walker method using the covariation based on spectral measure presented in Eq. (\ref{CV-def}). In this method, to estimate the unknown coefficient matrices ${\Theta}(v)$, where $v=1,\ldots,T$, we replace the elements of estimated normalized covariation matrices given in Eq. (\ref{system_est}) by estimated covariation matrices based spectral measure given in Section \ref{est_cov1}, i.e. we consider the following system of equations
\begin{equation}\label{system_est1}
\widehat{\mathbf{CV}}^v(1)= \widehat{\Theta}(v)\widehat{\mathbf{CV}}^{(v-1)}(0),
\end{equation}
for $v = 1,\ldots,T$. When the matrices $\widehat{\mathbf{CV}}^{(v-1)}(0)$ given in Eq. (\ref{system_est1}) are non-singular for $v=1,\ldots,T$, to calculate $\widehat{\mathbf{\Theta}}(v)$ we can use the following formula
\begin{equation}
    \widehat{\Theta}(v)=\widehat{\mathbf{CV}}^v(1){\widehat{\mathbf{CV}}^{(v-1)}(0)}^{-1}.
\end{equation}
Or
\begin{equation*}
\widehat{\Theta}(v)={SOLUTION}_1 \left(  \widehat{\mathbf{CV}}^{(v-1)}(0), \widehat{\mathbf{CV}}^v(1)\right), 
\end{equation*}
where $SOLUTION_1$ is an algorithm presented in Appendix B.

Solving the system of linear equations in the proposed approach produces biases in the
estimates. This is also demonstrated in the presented simulation study (see Section \ref{simul}).  Naturally, the bias is due to many similar or even overlapped quantities in the equations,
which leads to ill-conditioned matrices. One can try to reduce the bias applying the corrections for the obtained estimators. Such a  practice is often used especially for small trajectory lengths. In the considered case, the corrections can be obtained by the empirical analysis. However, the main goal of this paper is to demonstrate the idea of the new estimation methodology. Its further improvements   need to be investigated in the future study as a separate issue.

\section{Simulation study}\label{simul}
In this section, to verify the performance of the estimation procedures, we apply the proposed methods to two multidimensional PAR(1) time series models, i.e one is two-dimensional and second is three-dimensional versions of the PAR(1) time series models satisfying Eq. (\ref{par_p_m}) with $p=1$. Let us assume that $\{\mathbf{X}(t)\}=\{[(X_1(t),X_2(t))]'\}$ is two-dimensional PAR(1) time series model with two-dimensional symmetric $\alpha-$stable noise $\{\mathbf{Z}(t)\}=\{[Z_1(t),Z_2(t)]'\}$ with $1<\alpha<2$ and the following spectral measure 
\begin{align} \label{chosen_Gamma1}
	\Gamma=\gamma_1\delta\left(\left(z_1,z_2\right)\right)+\gamma_2\delta\left(\left(-z_1,-z_2\right)\right)+\gamma_3\delta\left(\left(-z_1,z_2\right)\right)+\gamma_4\delta\left(\left(z_1,-z_2\right)\right),
\end{align} 
where $\gamma_1=\gamma_2=\nu$ and $\gamma_3=\gamma_4=\xi$. For the $\alpha-$stable noise we take the following parameters: $z_1=1/2$, $z_2=\sqrt{3}/2$, $\nu=0.5$, and $\xi=0.2$.

Let us assume that $\{\mathbf{X}(t)\}=\{[(X_1(t),X_2(t),X_3(t))]'\}$ is three-dimensional PAR(1) time series model with the three-dimensional symmetric $\alpha-$stable noise $\{\mathbf{Z}(t)\}=\{[Z_1(t),Z_2(t),Z_3(t)]'\}$  with $1<\alpha<2$ and the following spectral measure 
\begin{multline} \label{chosen_Gamma2}
	\Gamma=\gamma_1\delta\left(\left(z_1,z_2,z_3\right)\right)+\gamma_2\delta\left(\left(-z_1,-z_2,-z_3\right)\right)+\gamma_3\delta\left(\left(-z_1,z_2,z_3\right)\right)+\gamma_4\delta\left(\left(z_1,-z_2,-z_3\right)\right)\\+\gamma_5\delta\left(\left(z_1,-z_2,z_3\right)\right)+\gamma_6\delta\left(\left(-z_1,z_2,-z_3\right)\right)+\gamma_7\delta\left(\left(z_1,z_2,-z_3\right)\right)+\gamma_8\delta\left(\left(-z_1,-z_2,z_3\right)\right),
\end{multline}
where $\gamma_1=\gamma_2=\nu_1$, $\gamma_3=\gamma_4=\xi_1$, $\gamma_5=\gamma_6=\nu_2$ and $\gamma_7=\gamma_4=\xi_2$. For the $\alpha-$stable noise we take the following parameters: $z_1=1/2$, $z_2=1/2$, $z_3 =\sqrt{2}/2$, $\nu_1=0.1$, $\xi_1=0.2$, $\nu_2=0.3$ and $\xi_2=0.5$.

We mention here that an exact method for simulating the multidimensional $\alpha$-stable random vectors is presented in \citep{Nolan_simulations}. In the simulation study presented here  two PAR(1) models, denoted as Model 1 (two-dimensional PAR(1) model) and Model 2 (three-dimensional PAR(1) model), that correspond to $T=3$ and $T = 2$, respectively, with the following coefficient matrices 
{\begin{equation}
	\Theta(1)= \left[
	\begin{array}{cc}
	0.5 & 0.1 \\
	-0.6 & 0.4 \\	\end{array}
	\right], \quad
	\Theta(2)= \left[
	\begin{array}{cc}
	0.8 & -0.1 \\
	0.3 & 0.7 \\	\end{array}
	\right], \quad
	\Theta(3)= \left[
	\begin{array}{cc}
	0.1 & -0.4 \\
	-0.5 & 0.3 \\	\end{array}
	\right]
	\label{Theta_1}
\end{equation}}
for Model 1, and 
{\begin{equation}
	\Theta(1)= \left[
	\begin{array}{ccc}
	0.8 &-0.2 & 0.7 \\
	0.1 & 0.5 & -0.6 \\	
	0.4 & 0.3 & -0.1 \\\end{array}
	\right], \quad
	\Theta(2)= \left[
	\begin{array}{ccc}
	0.4 & -0.1 & 0.3 \\
	0.5 & -0.2 & 0.4\\
	-0.3 & 0.8 &-0.6	\end{array}
	\right]
	\label{Theta_2}
\end{equation}}
for Model 2. Sample trajectories of length $L=1000$ of the considered time series are given in Figs. \ref{fig1:} and \ref{fig4:} for Model 1 and Model 2 respectively. Next, we estimate PAR coefficient matrices for Model 1 and Model 2 using Y-W methods based on moment-based covariation (Y-W-CV method) and  spectral measure-based covariation (Y-W-T method)  given in Section \ref{est_method}. The results are obtained by both estimation method using the Monte-Carlo simulations, namely, for each model, we generate $M = 1000$ trajectories of length $L= 1000$.  Also, we compare the estimated coefficient matrices using Y-W-CV method with using the Y-W-T method. In Figs. \ref{fig2:} and \ref{fig5:}, we present the boxplots of the estimated coefficient matrices using both methods for Model 1 and Model 2, respectively. 
As one can see, the medians of the values taken by estimators using both methods are close to the real values of the coefficient matrices. We can observe that the Y-W-CV-based estimators are narrower in comparison to Y-W-T ones. Thus, the Y-W-CV method is preferable here and will be further used for real data analysis presented in the next section. From Figs. \ref{fig2:} and \ref{fig5:} one can suppose that the variances for the coefficients from different periods are different. There is no reason for that  as all coefficients are estimated in the same way. This is due
to the arrangement of the figures (i.e. application of different scales on the subplots) applied to compare the results between two estimation methods.

\begin{figure} [h!]
  \centering
  \includegraphics[scale=0.35]{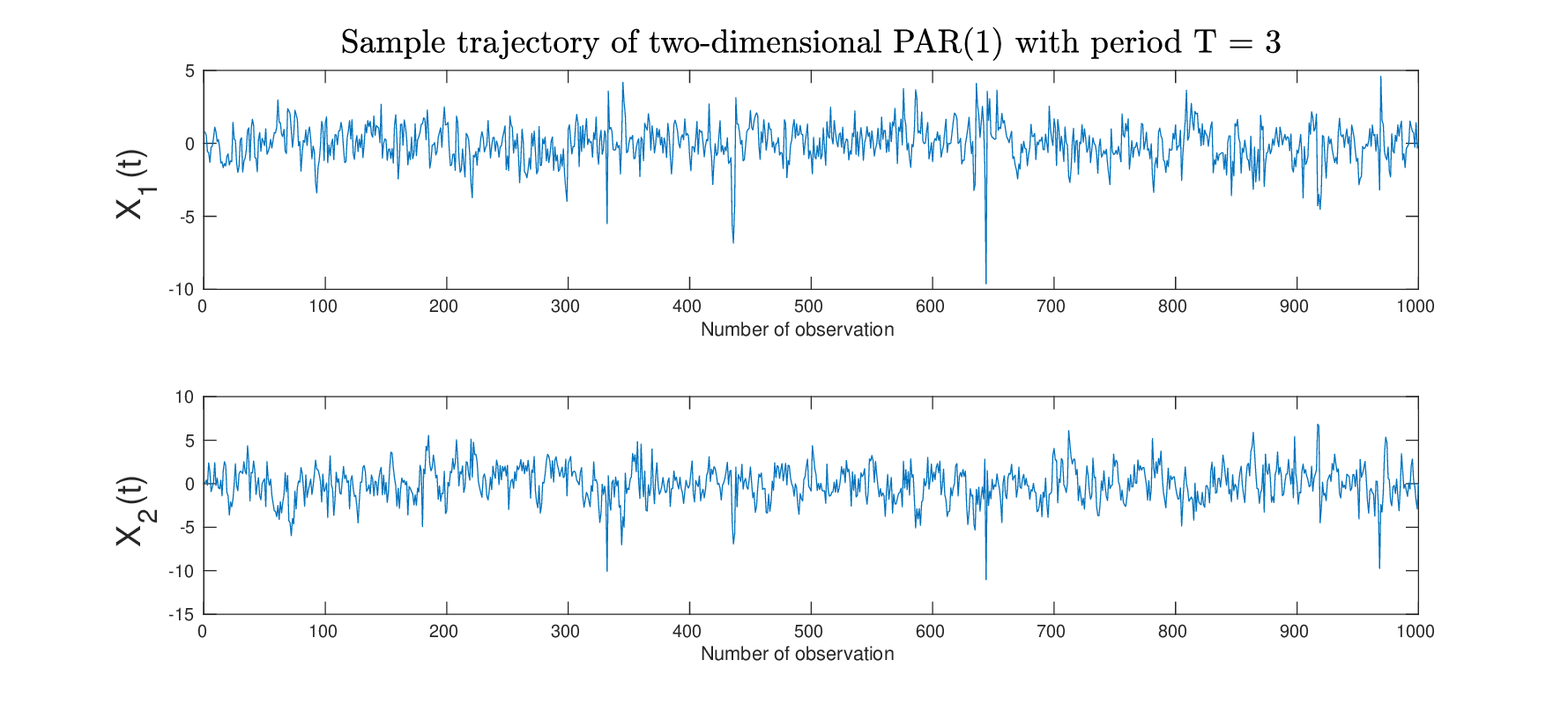}
  \caption{{Sample trajectories of the two-dimensional $\alpha-$stable PAR(1) time series of length $L =1000$ with $T = 3$ - Model 1. The coefficient matrices are given in Eq. (\ref{Theta_1}), $\alpha=1.8$ and the spectral measure is defined in Eq. (\ref{chosen_Gamma1}).}}
  \label{fig1:}
\end{figure}

\begin{figure}[h!]
  \centering
  \includegraphics[scale=0.26]{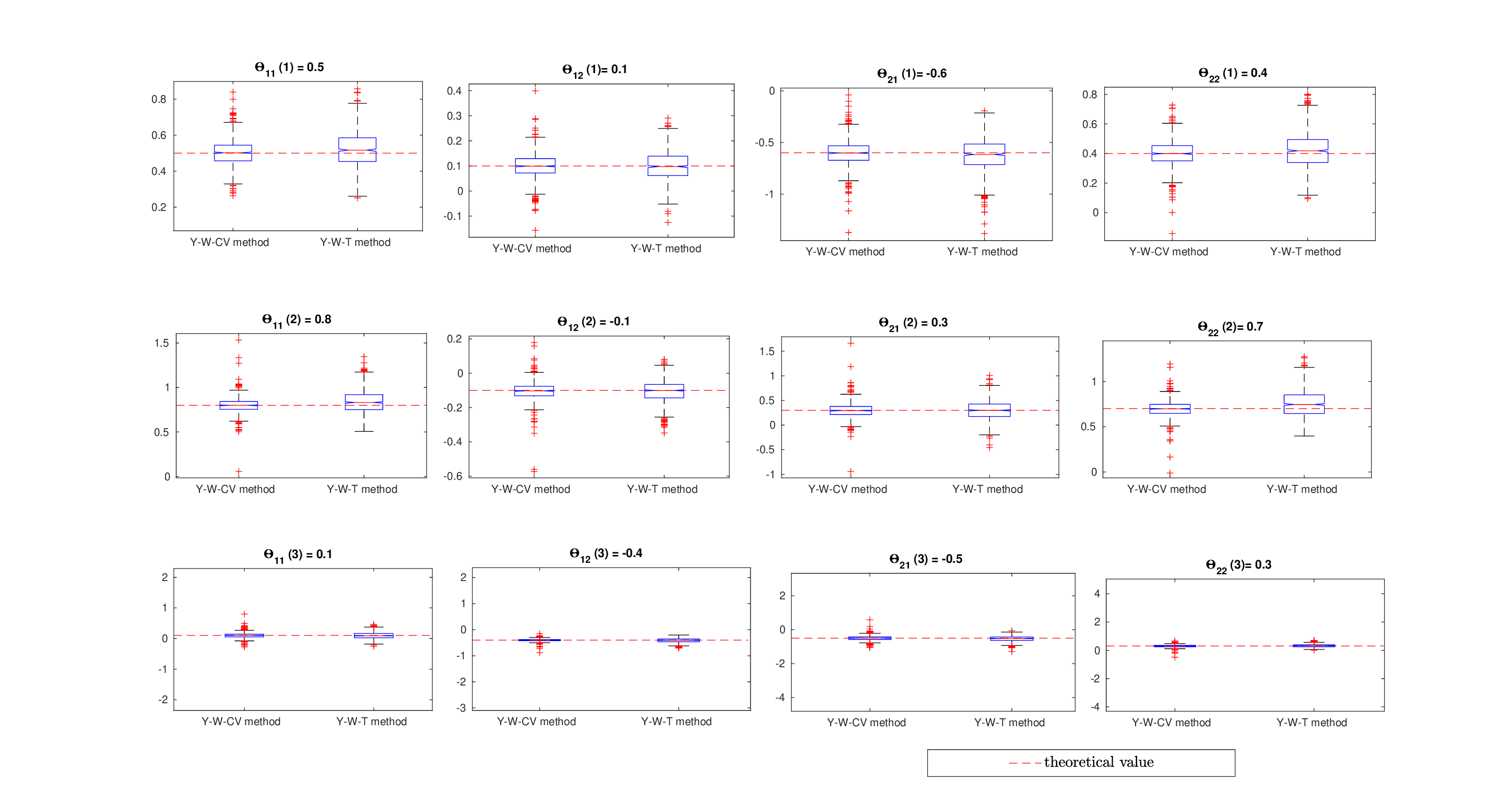}
  \caption{The boxplots of estimated coefficient matrices for the two-dimensional $\alpha-$stable PAR(1) time series of length length $L =1000$ with $T = 3$ - Model 1. The coefficient matrices are given in Eq. (\ref{Theta_1}), $\alpha=1.8$ and the spectral measure is defined in Eq. (\ref{chosen_Gamma1}). The boxplots on left hand side of all panels describe the estimated coefficient matrices values using Y-W-CV method while the right boxplots using the Y-W-T method. The results are obtained by using Monte-Carlo simulations with the number of repetitions equal to $M=1000$.} 
  \label{fig2:}
\end{figure}

\begin{figure}[h!]
  \centering
  \includegraphics[scale=0.26]{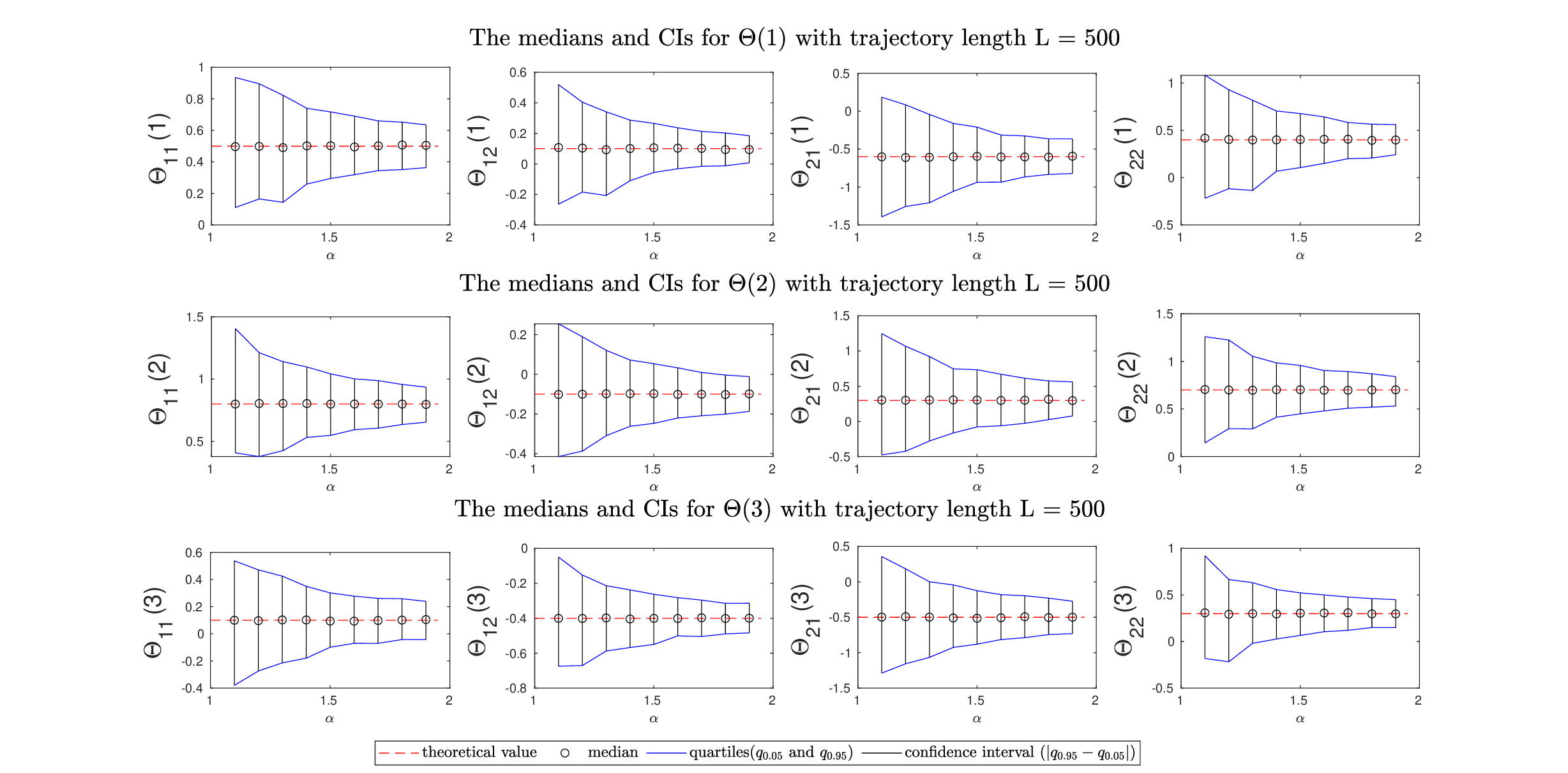}
      \caption{{The graph presenting the results (medians and 90\% empirical confidence intervals using empirical quantiles function) for the Y-W-CV estimation method applied to the two-dimensional $\alpha-$stable PAR(1) time series of length $L =500$ with $T=3$ - Model 1. The coefficient matrices are given in Eq. (\ref{Theta_1}), the spectral measure is defined in Eq. (\ref{chosen_Gamma1}) and $\alpha$ is changing from $\alpha=1.1$ to $\alpha=1.9$. The results are obtained by using Monte-Carlo simulations with the number of repetitions equal to $M=1000$.}}
  \label{fig7:}
\end{figure}
\begin{figure}[h!]
  \centering
  \includegraphics[scale=0.26]{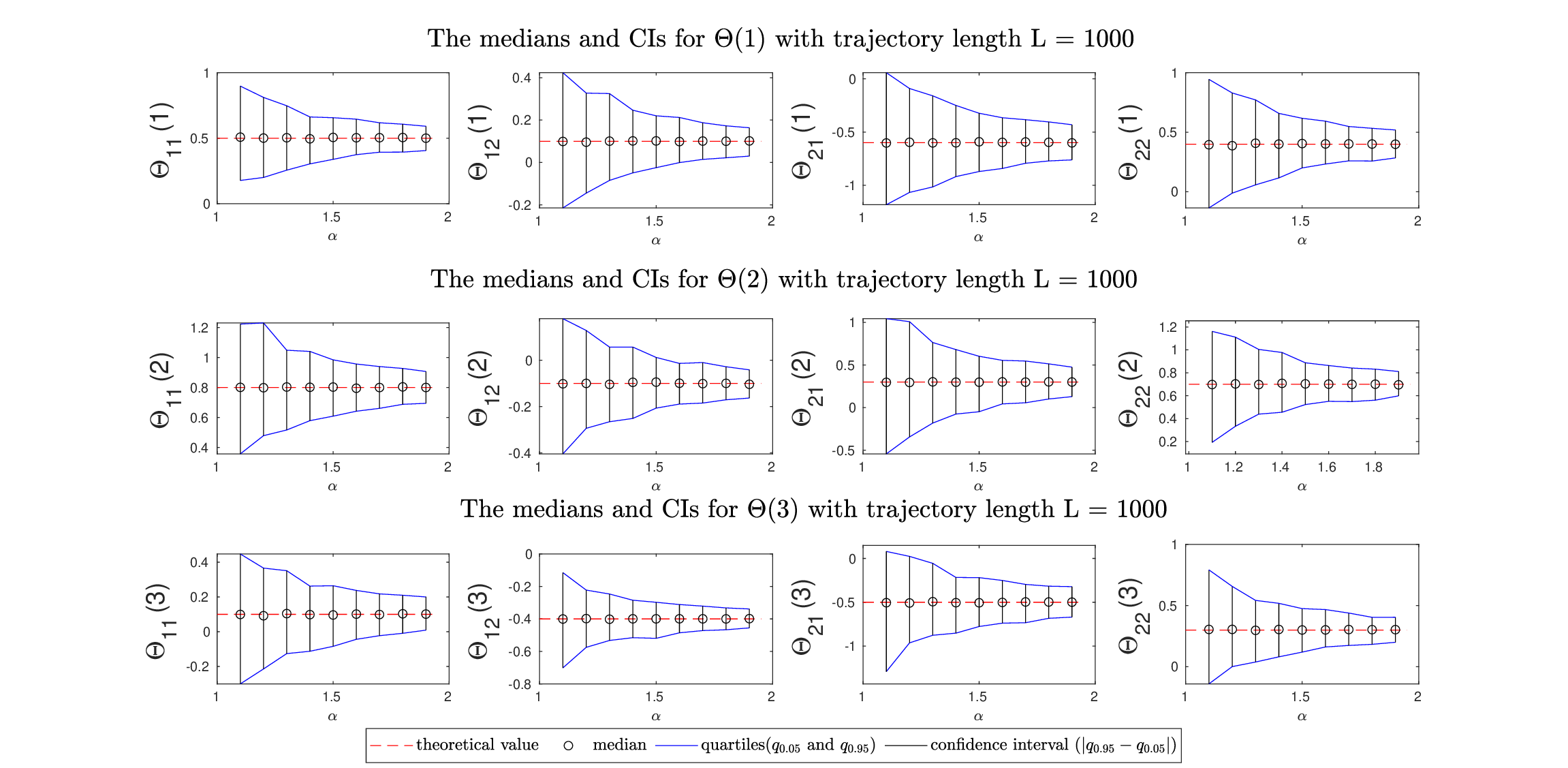}
    \caption{{The graph presenting the results (medians and 90\% empirical confidence intervals using empirical quantiles function) for the Y-W-CV estimation method applied to the two-dimensional $\alpha-$stable PAR(1) time series of length $L =1000$ with $T=3$ - Model 1. The coefficient matrices are given in Eq. (\ref{Theta_1}), the spectral measure is defined in Eq. (\ref{chosen_Gamma1}) and $\alpha$ is changing from $\alpha=1.1$ to $\alpha=1.9$. The results are obtained by using Monte-Carlo simulations with the number of repetitions equal to $M=1000$.}}
  \label{fig8:}
\end{figure}
\begin{figure}[h!]
  \centering
  \includegraphics[scale=0.26]{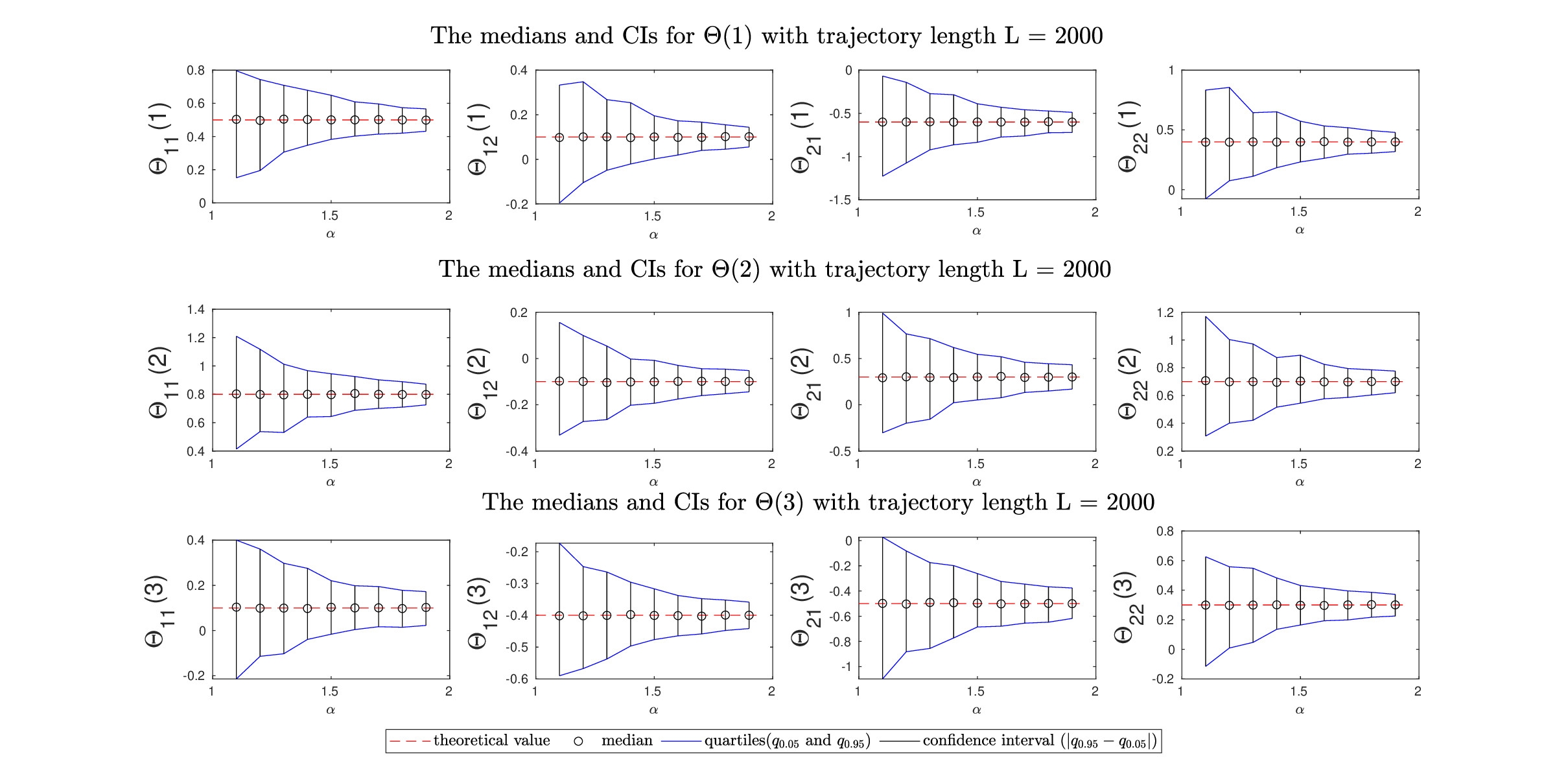}
    \caption{{The graph presenting the results (medians and 90\% empirical confidence intervals using empirical quantiles function) for the Y-W-CV estimation method applied to the two-dimensional $\alpha-$stable PAR(1) time series of length $L =2000$ with $T=3$ - Model 1. The coefficient matrices are given in Eq. (\ref{Theta_1}), the spectral measure is defined in Eq. (\ref{chosen_Gamma1}) and $\alpha$ is changing from $\alpha=1.1$ to $\alpha=1.9$. The results are obtained by using Monte-Carlo simulations with the number of repetitions equal to $M=1000$.}}
  \label{fig9:}
\end{figure}

\begin{figure} [h!]
  \centering
  \includegraphics[scale=0.35]{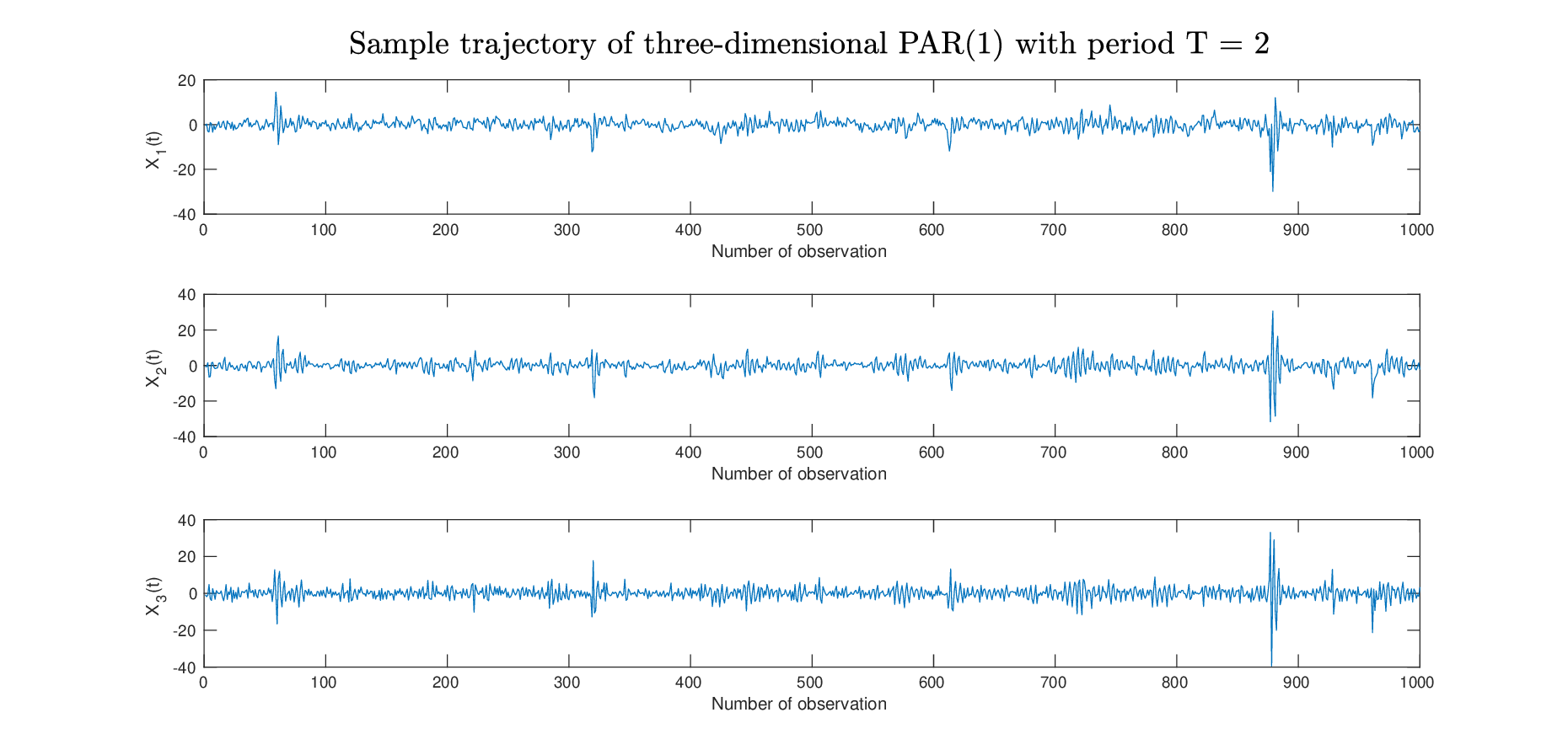}
  \caption{Sample trajectories of the three-dimensional $\alpha-$stable PAR(1) time series of length $L =1000$ with $T = 2$ - Model 2. The coefficient matrices are given in Eq. (\ref{Theta_2}), $\alpha=1.8$ and the spectral measure is defined in Eq. (\ref{chosen_Gamma2}).}
  \label{fig4:}
\end{figure}

\begin{figure}[h!]
  \centering
  \includegraphics[scale=0.26]{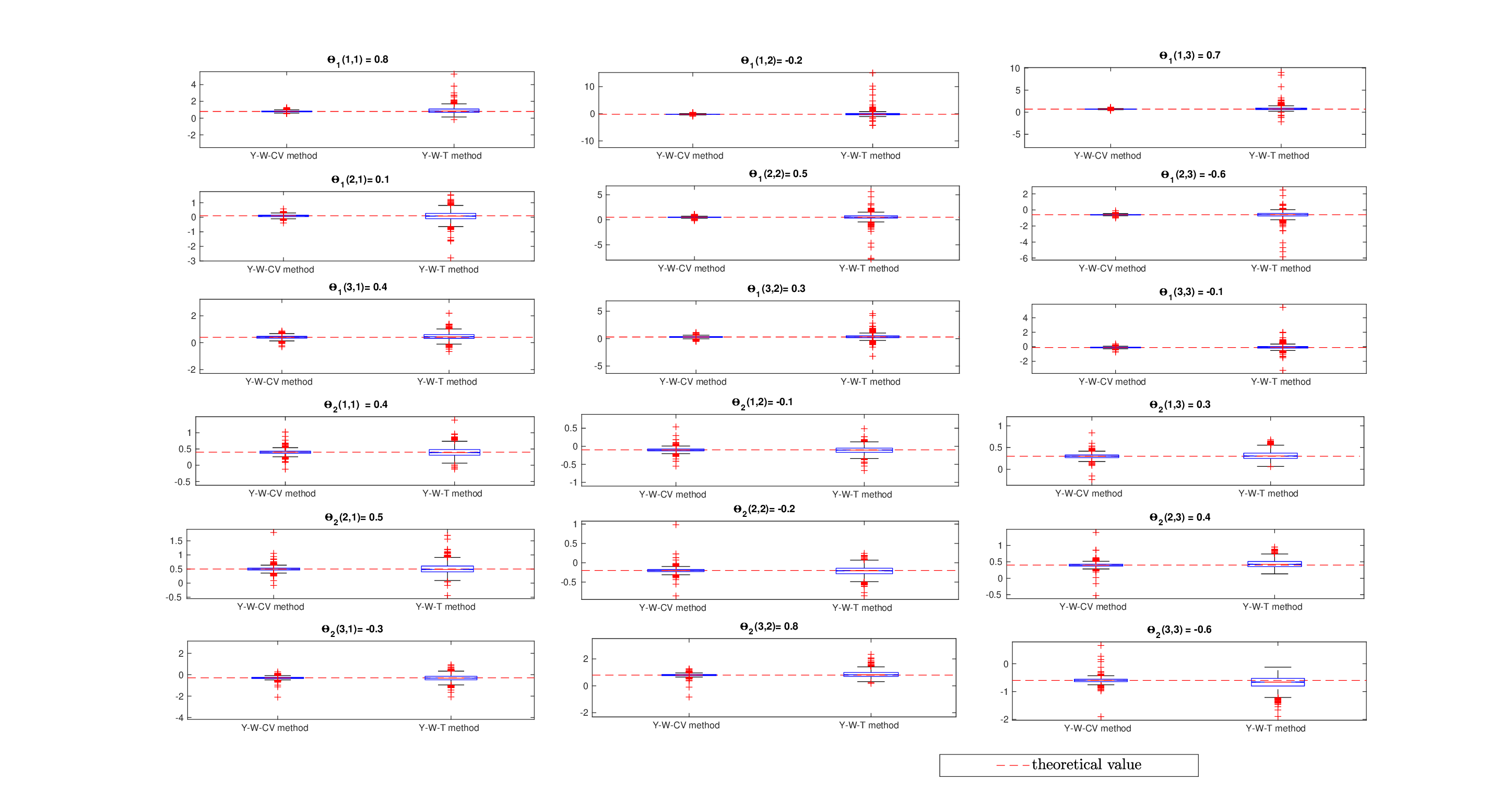}
  \caption{The boxplots of estimated coefficient matrices for the three-dimensional $\alpha-$stable PAR(1) time series of length length $L =1000$ with $T = 2$ - Model 2. The coefficient matrices are given in Eq. (\ref{Theta_2}), $\alpha=1.8$ and the spectral measure is defined in Eq. (\ref{chosen_Gamma2}). The boxplots on left hand side of all panels describe the estimated coefficient matrices values using Y-W-CV method while the right boxplots using the Y-W-T method. The results are obtained by using Monte-Carlo simulations with the number of repetitions equal to $M=1000$. }
  \label{fig5:}
\end{figure}

\begin{figure}[h!]
  \centering
  \includegraphics[scale=0.26]{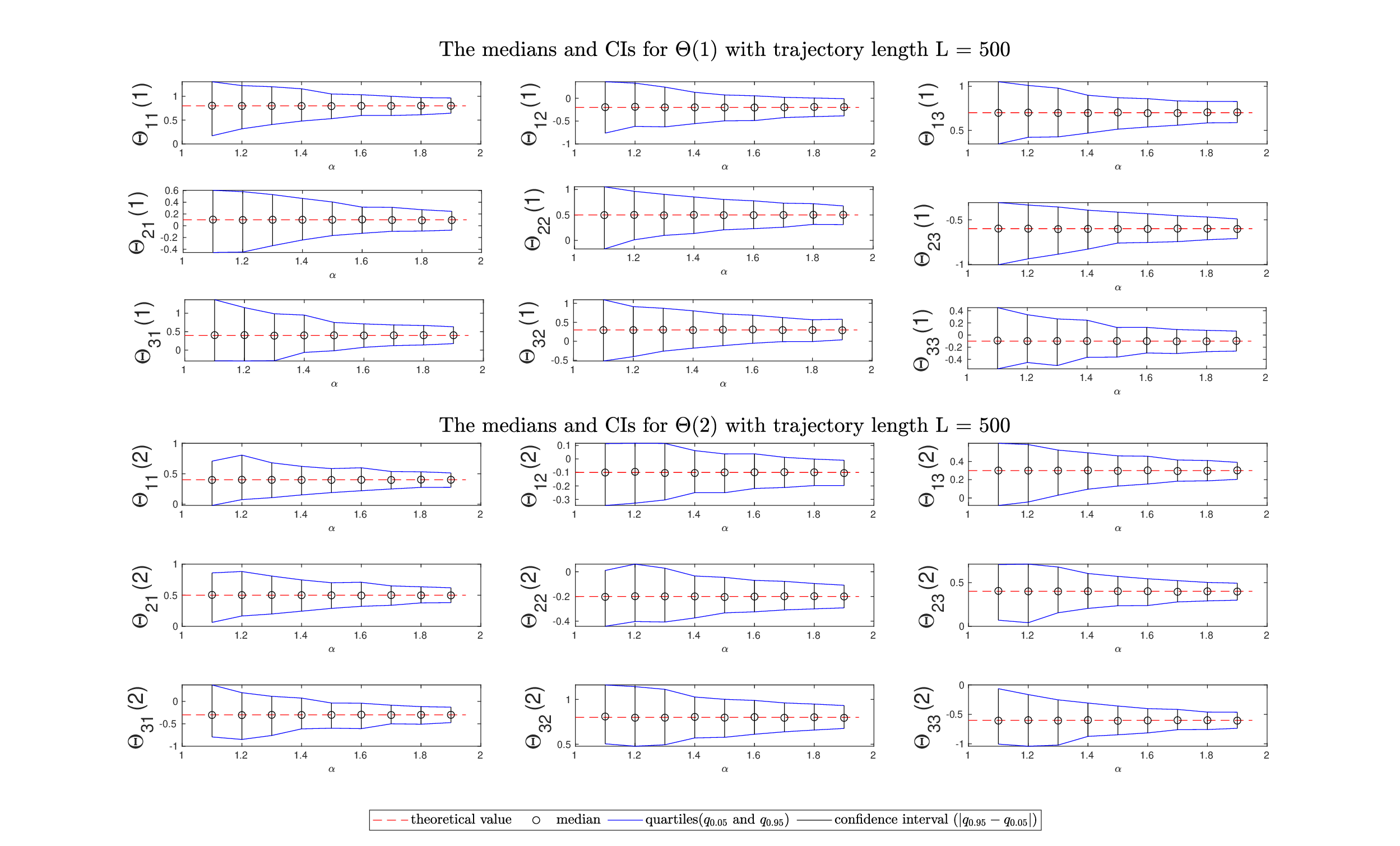}
    \caption{{The graph presenting the results (medians and 90\% empirical confidence intervals using empirical quantiles function) for the Y-W-CV estimation method applied to the three-dimensional $\alpha-$stable PAR(1) time series of length $L =500$ with $T=2$ - Model 2. The coefficient matrices are given in Eq. (\ref{Theta_2}), the spectral measure is defined in Eq. (\ref{chosen_Gamma2}) and $\alpha$ is changing from $\alpha=1.1$ to $\alpha=1.9$. The results are obtained by using Monte-Carlo simulations with the number of repetitions equal to $M=1000$.}}
  \label{fig10:}
\end{figure}
\begin{figure}[h!]
  \centering
  \includegraphics[scale=0.26]{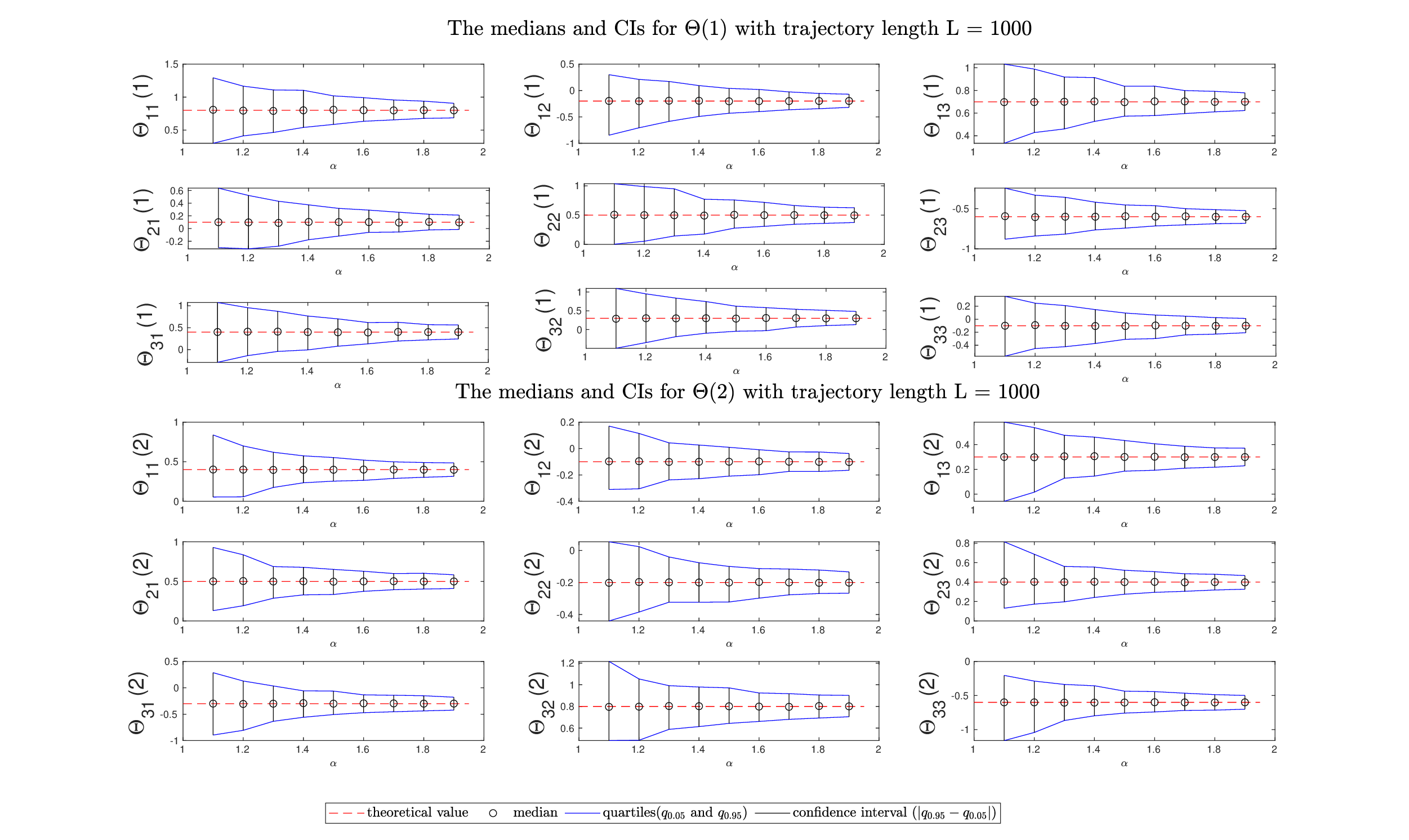}
  \caption{{The graph presenting the results (medians and 90\% empirical confidence intervals using empirical quantiles function) for the Y-W-CV estimation method applied to the three-dimensional $\alpha-$stable PAR(1) time series of length $L =1000$ with $T=2$ - Model 2. The coefficient matrices are given in Eq. (\ref{Theta_2}), the spectral measure is defined in Eq. (\ref{chosen_Gamma2}) and $\alpha$ is changing from $\alpha=1.1$ to $\alpha=1.9$. The results are obtained by using Monte-Carlo simulations with the number of repetitions equal to $M=1000$.}}
  \label{fig11:}
\end{figure}
\begin{figure}[h!]
  \centering
  \includegraphics[scale=0.26]{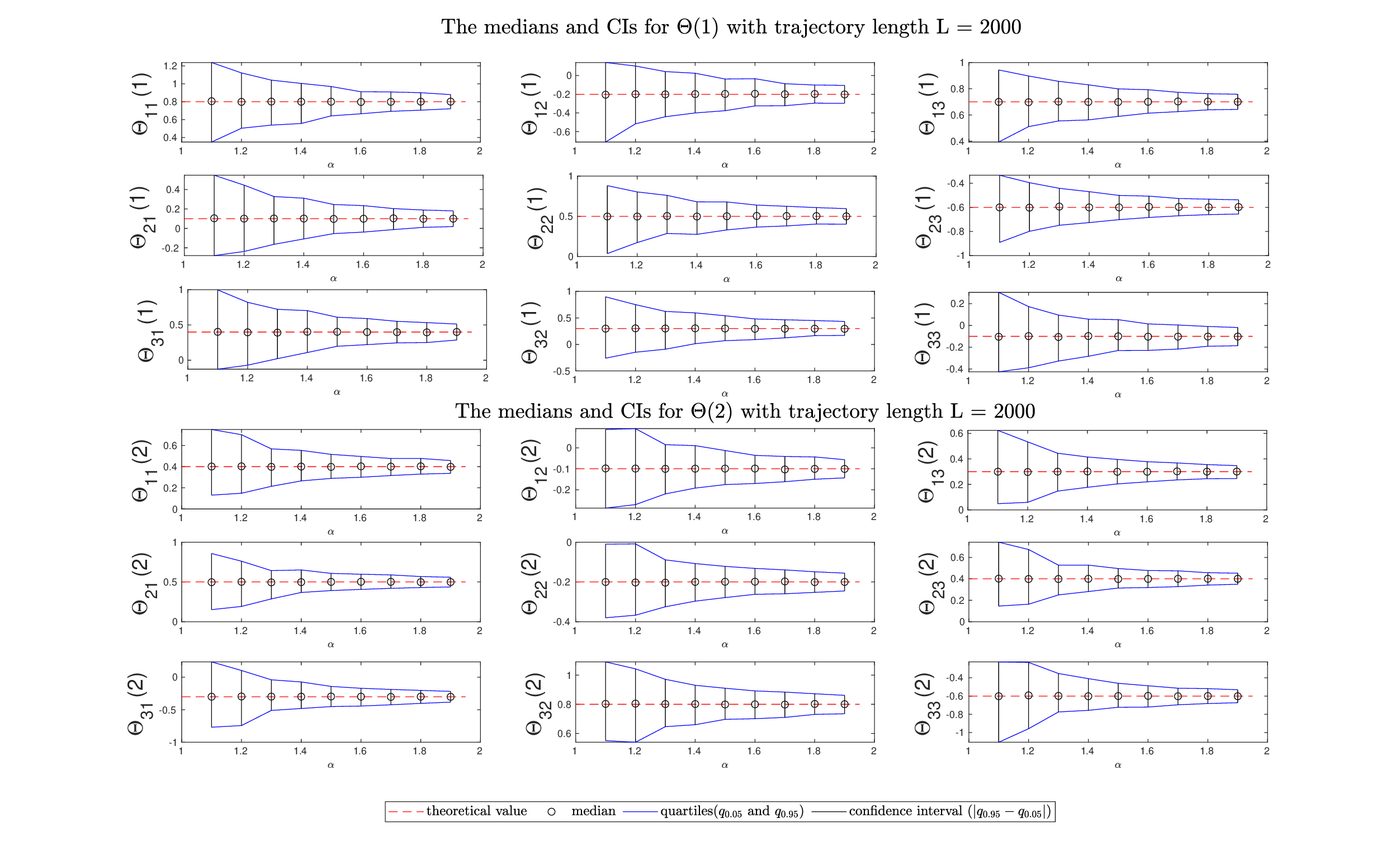}
      \caption{{The graph presenting the results (medians and 90\% empirical confidence intervals using empirical quantiles function) for the Y-W-CV estimation method applied to the three-dimensional $\alpha-$stable PAR(1) time series of length $L =2000$ with $T=2$ - Model 2. The coefficient matrices are given in Eq. (\ref{Theta_2}), the spectral measure is defined in Eq. (\ref{chosen_Gamma2}) and $\alpha$ is changing from $\alpha=1.1$ to $\alpha=1.9$. The results are obtained by using Monte-Carlo simulations with the number of repetitions equal to $M=1000$.}}
  \label{fig12:}
\end{figure}
For the considered time series models, we want to evaluate the accuracy of the proposed Y-W-CV estimation procedure for different values of the stability index and different trajectory lengths. Therefore, for both models we generate $M = 1000$ sample trajectories of the lengths $L$ with $\alpha$ changing from $\alpha = 1.1$ to $\alpha = 1.9$. Based on the simulated trajectories, we estimate the unknown coefficients of the model using the proposed Y-W-CV method. For the outcomes corresponding to each value of $\alpha$ we calculate the median and the quantiles of order $0.05$ and $0.95$ ($90\%$ empirical confidence intervals using empirical quantiles function). For Model 1 and Model 2, the results are presented in Figs. \ref{fig7:}-\ref{fig9:} and Figs. \ref{fig10:}-\ref{fig12:}, respectively, for the trajectories of length $L = 500$, $L=1000$ and $L=2000$. As one can observe, for both models the medians of the values taken by the estimated coefficients are close to theoretical ones and the length of the calculated empirical confidence intervals decreases if the values of $\alpha$ and $L$ increase. 

\section{Real data analysis} \label{realdata} 
In this section, the application of the considered model to the real data is presented. The analyzed dataset comes from the Nord Pool power day-ahead market for the Swedish  {SE2} area and contains two hourly time series -- prices \citep{data_prices} and volumes bought \citep{data_volumes} --  {from 31st May 2018 through 29th July 2018 \footnote{ {The data were available and downloaded in September 2020.}}}. In both trajectories, presented in Fig. \ref{fig:data_raw}, there are  {1440} observations. The data are taken from the energy market and thus, we expect the periodic behavior related to the $T=24$, as we are considering hourly data and thus the periodic model is a natural choice. On the other hand, it is obvious that the prices and volumes bought are related. Thus, the two-dimensional model is proposed. Finally, the data are expected to be non-Gaussian (large observations are visible in both datasets). Thus, the $\alpha-$stable distribution seems to be justified for the distribution of the residual vector. From each trajectory, the fitted first-degree polynomial and the periodic mean are subtracted. To the data without these components, presented in Fig. \ref{fig:data_dt}, the two-dimensional $\alpha$-stable PAR(1) model with period $T$ is now fitted, using the Y-W-CV method introduced in Section \ref{est_method}. The estimated values for all $2\cdot 2 \cdot 24 = 96$ coefficients are presented in Table \ref{tab:ests} in Appendix C.

\begin{figure}[h!]
  \centering
  \includegraphics[scale=0.37]{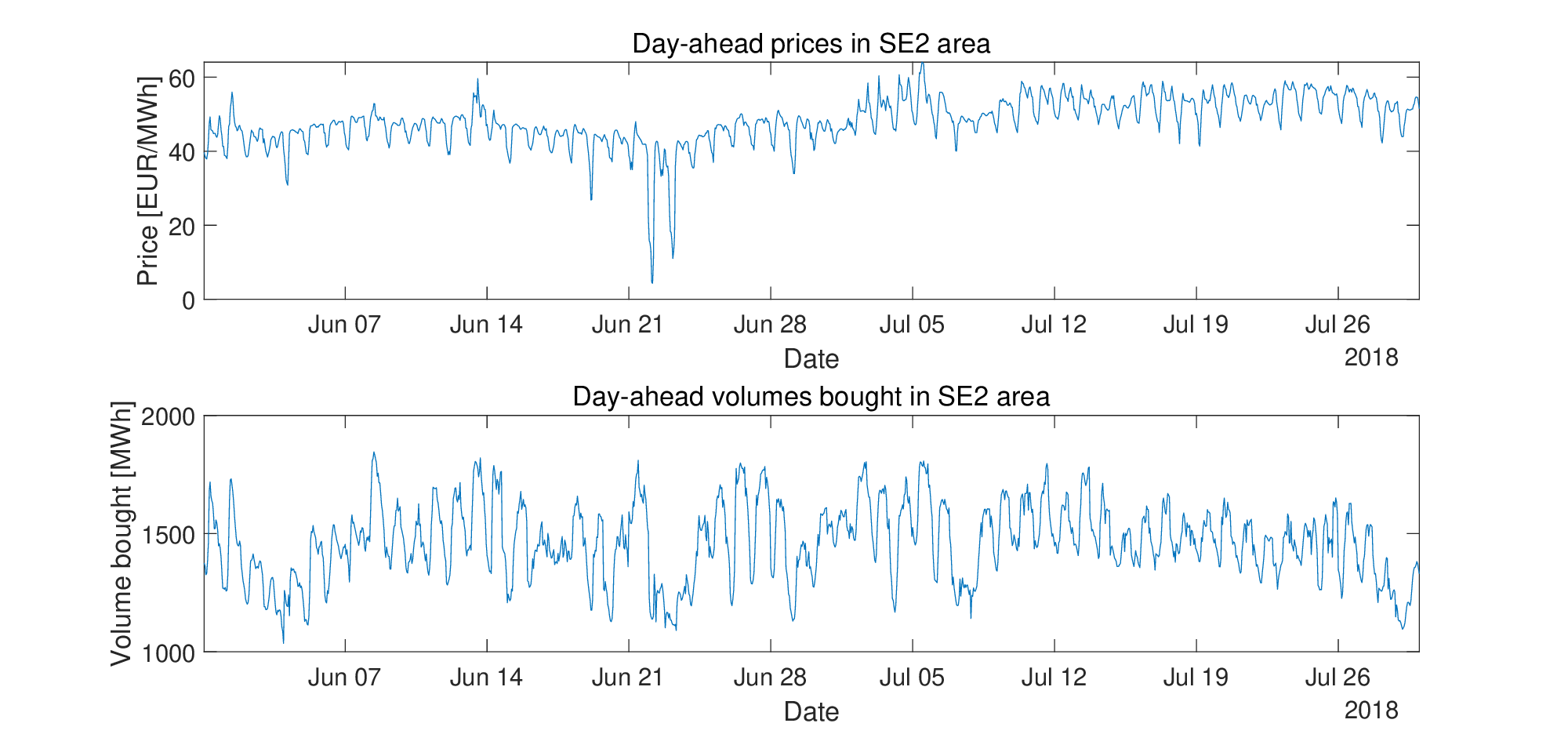}
  \caption{ {The analyzed dataset: hourly prices (top panel) and volumes bought (bottom panel) in Swedish SE2 area in the Nord Pool power day-ahead market from 31st May 2018 through 29th July 2018.}}
  \label{fig:data_raw}
\end{figure}

\begin{figure}[h!]
  \centering
  \includegraphics[scale=0.37]{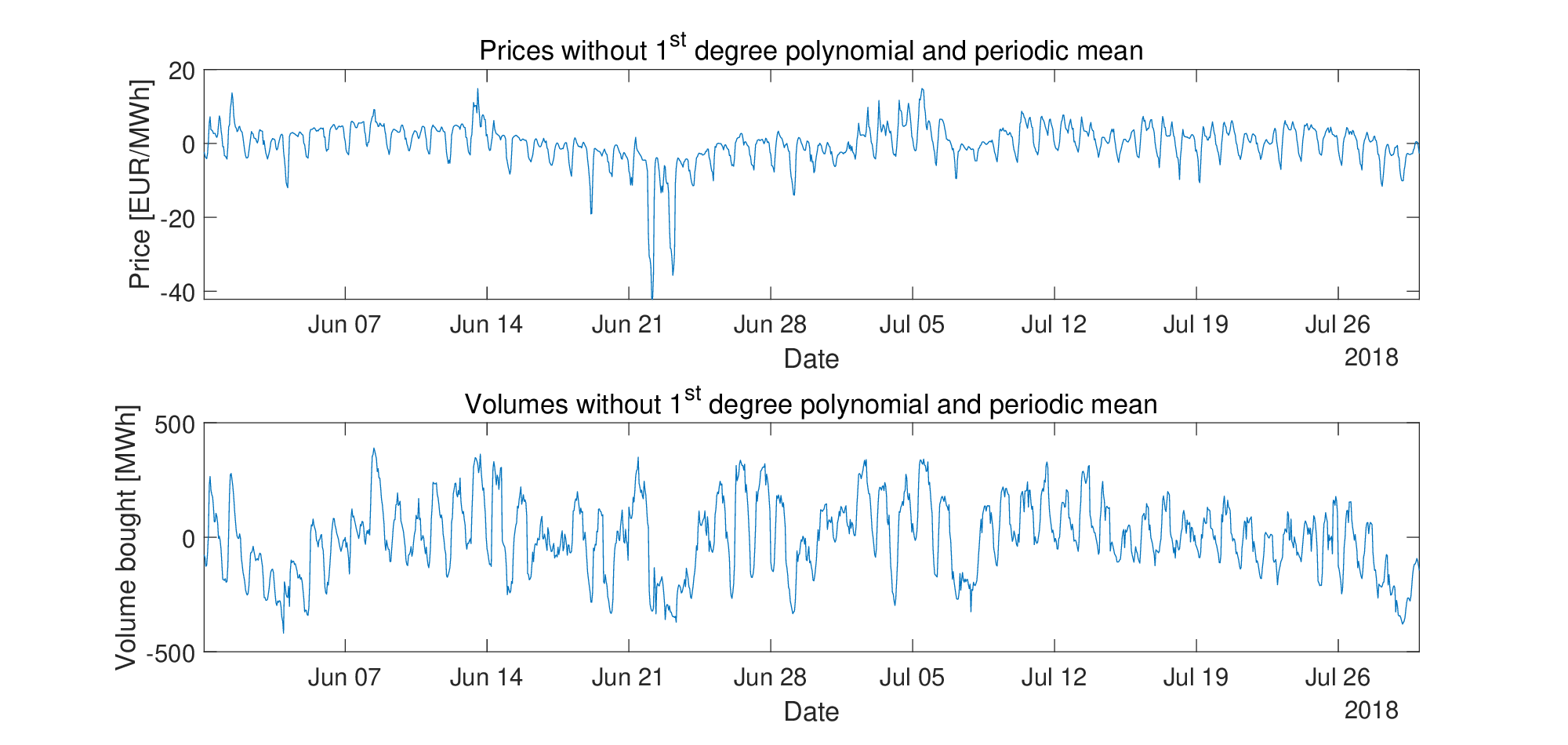}
  \caption{ {The analyzed prices (top panel) and volumes (bottom panel) series with the fitted first-degree polynomial and the periodic mean (with period $T=24$) subtracted.}}
  \label{fig:data_dt}
\end{figure}

The obtained residual series are illustrated in Fig. \ref{fig:data_residuals}. First, we check if each residual vector consists of independent observations. As we assume the residuals are $\alpha$-stable distributed, the auto-covariation measure is used. The top and middle panels of Figs. \ref{fig:data_resncv} present the estimated normalized auto-covariations for each residual vector. In both cases, there are no signs of a significant interdependence. However, one can see that the situation is not perfect. Especially for the price data the normalized auto-covariation takes no zero values for $h\neq 0$. The data are used here to demonstrate the possible applications of the presented methodology, and probably the proposed model is not the optimal one. However, it can be used as a preliminary description of the specific behavior visible in the data. Moreover, as one can see in the bottom panel of Fig. \ref{fig:data_resncv}, the estimated values of normalized cross-covariation indicate  {that there is no dependence between the residual vector components.} This is also confirmed by the estimated spectral measure presented in Fig. \ref{fig:spect2d} in Appendix D.

\begin{figure}[h!]
  \centering
  \includegraphics[scale=0.37]{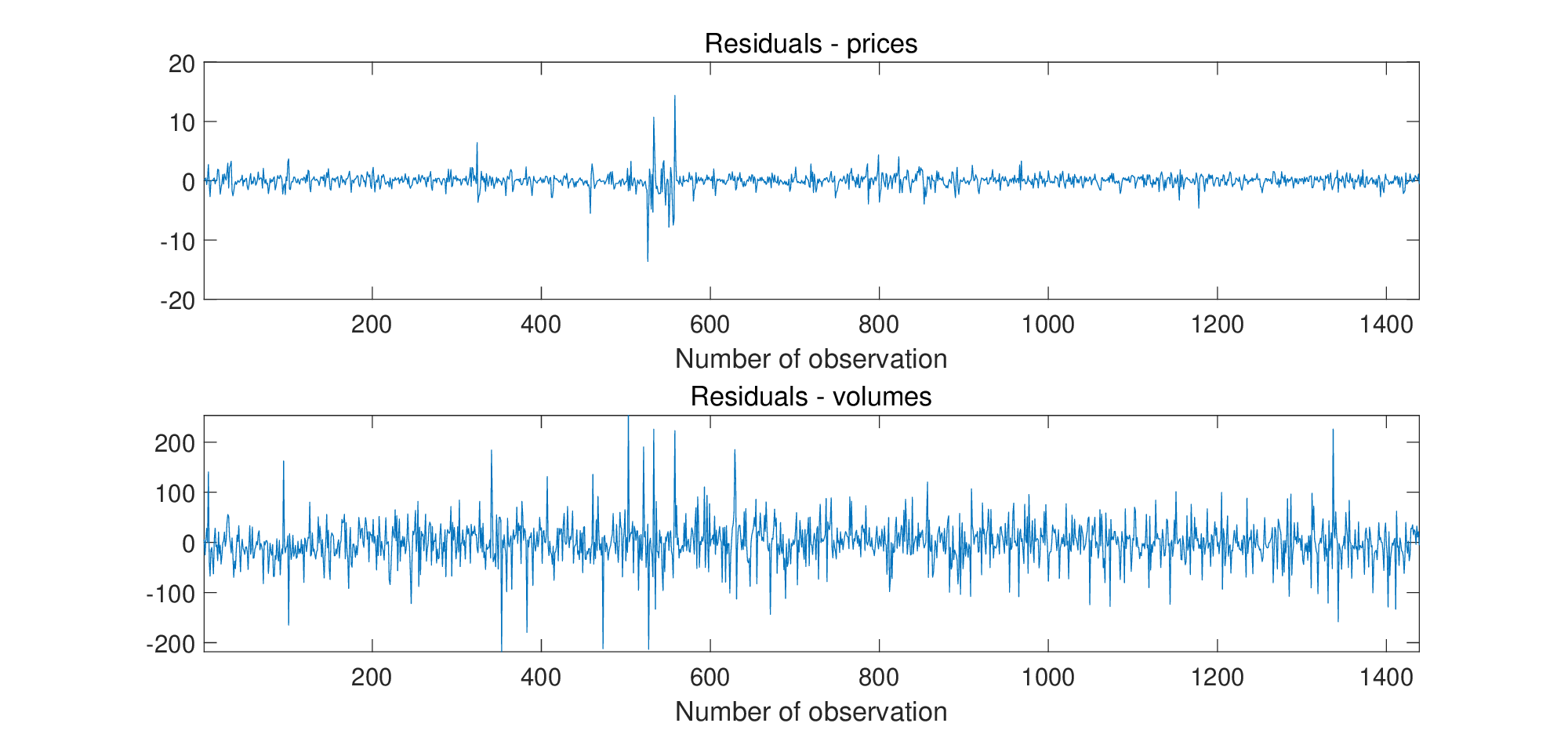}
  \caption{ {The obtained residuals for prices (top panel) and volumes (bottom panel) series.}}
  \label{fig:data_residuals}
\end{figure}

\begin{figure}[h!]
  \centering
  \includegraphics[scale=0.37]{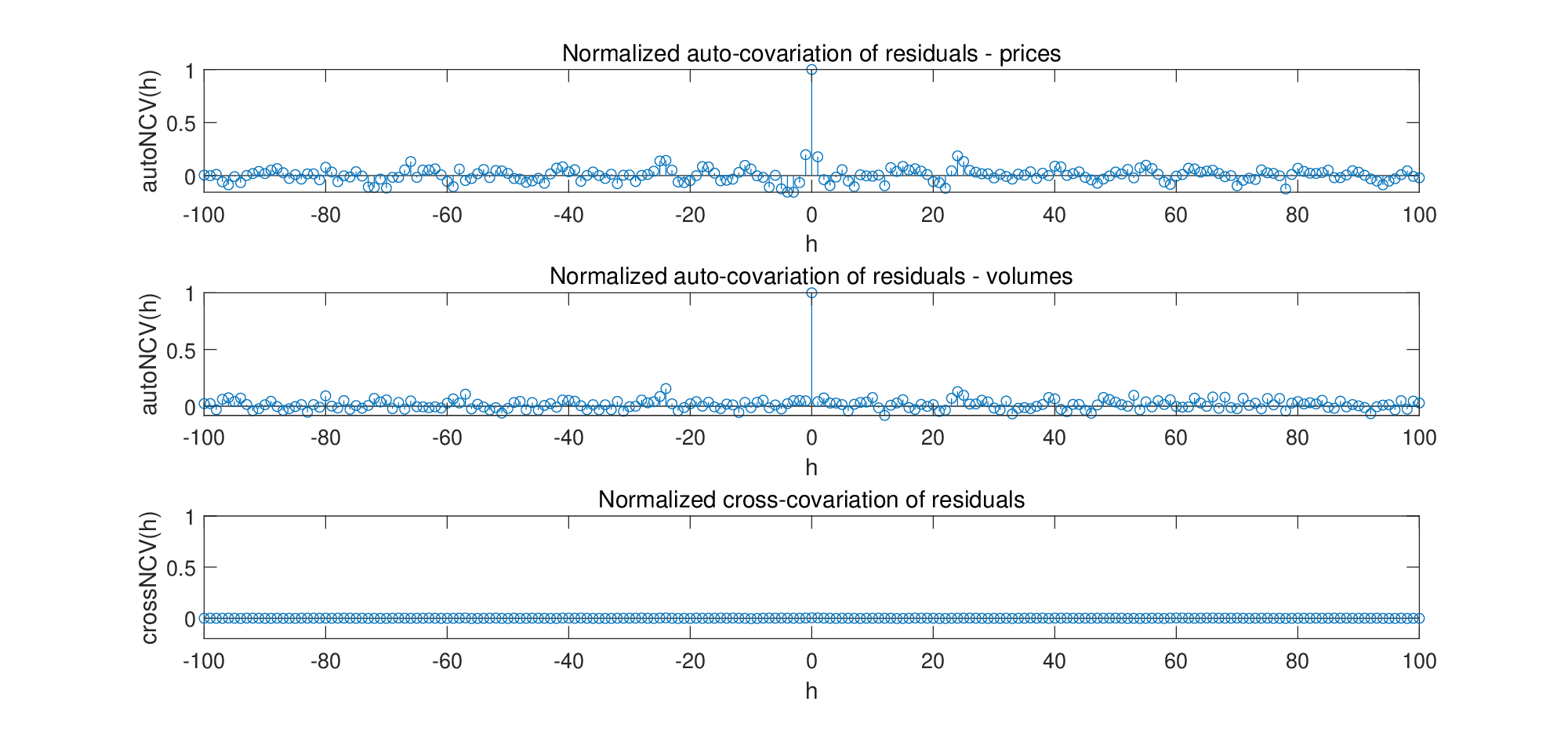}
  \caption{ {Top panel: normalized auto-covariation of prices' residual series, middle panel: normalized auto-covariation of volumes' residual series, bottom panel: normalized cross-covariation of residuals. }}
  \label{fig:data_resncv}
\end{figure}

Next, we analyze the distribution of the residuals and check whether they can be considered as a sample of a two-dimensional $S\alpha S$ vector. Separately, for each residual series, the parameters of the $\alpha$-stable distribution are estimated using the McCulloch method \citep{fin3}. To confirm that both vectors can be treated as $\alpha$-stable distributed samples, we perform the Monte Carlo-based Anderson-Darling (A-D) test with 1000 simulations \citep{rweron}. The results are presented in Table \ref{tab:data}. From the obtained empirical p-values, one can conclude that there is no reason to reject the hypothesis about the $\alpha$-stable distribution in both cases. Table \ref{tab:data} also contains the estimated values of the parameter $\alpha$ for both residual vector. One can see that they are close to each other hence one can claim the residuals constitute two-dimensional vector from $\alpha-$stable distribution. The estimated spectral measure corresponding to the residual vector (after the normalization by the estimated scale parameters) is presented in Fig. \ref{fig:spect2d} given in Appendix D. The estimated point masses and weight suggest that the components of the two-dimensional $\alpha$-stable vector can be  {considered as independent. Moreover, the estimated spectral measure indicates that the measure is symmetric.} 

\begin{table}[!h]
	\small\centering
	\caption{ {Empirical p-values of the Anderson-Darling test based on 1000 Monte Carlo simulations and estimated values of stability index $\alpha$ for residual series.}} 
	\vspace{5pt} \centering
\begin{tabular}{||c||c||c||}
		\hline 
		residuals &  p-value of A-D test & estimated value of $\alpha$  \\
		\hline
		\hline
		prices &  0.2060 & 1.4095    \\ 
		\hline
		volumes &  0.1720 & 1.4103 \\ 
		\hline

	\end{tabular}
\label{tab:data}		
\end{table}

Finally, we have generated the quantile lines, using the following procedure. The 5000 trajectories of the two-dimensional $\alpha$-stable PAR(1) model with period $T=24$ and estimated parameters are simulated. The residuals are simulated from the two-dimensional $\alpha-$stable distribution with the estimated spectral measure and $\alpha$ parameter equal to the mean of the estimated stability indices obtained for  components of the residual vector.  Using these trajectories, we construct the quantile lines, taking quantiles of order $q = 0.1,\,0.5,\,0.9$. To compare them with the original dataset, the previously subtracted deterministic components (first-degree polynomials and periodic means) are added to the simulated series. This comparison, illustrated in Fig. \ref{fig:data_qlines}, confirms that the considered model can be used for description of the  examined data. Moreover, in Fig. \ref{fig:qlines_1step} we present also the one-step ahead conditional quantiles constructed based on the fitted model. This plot also confirms that the model is appropriate for the considered real data.

\begin{figure}[h!]
  \centering
  \includegraphics[scale=0.37]{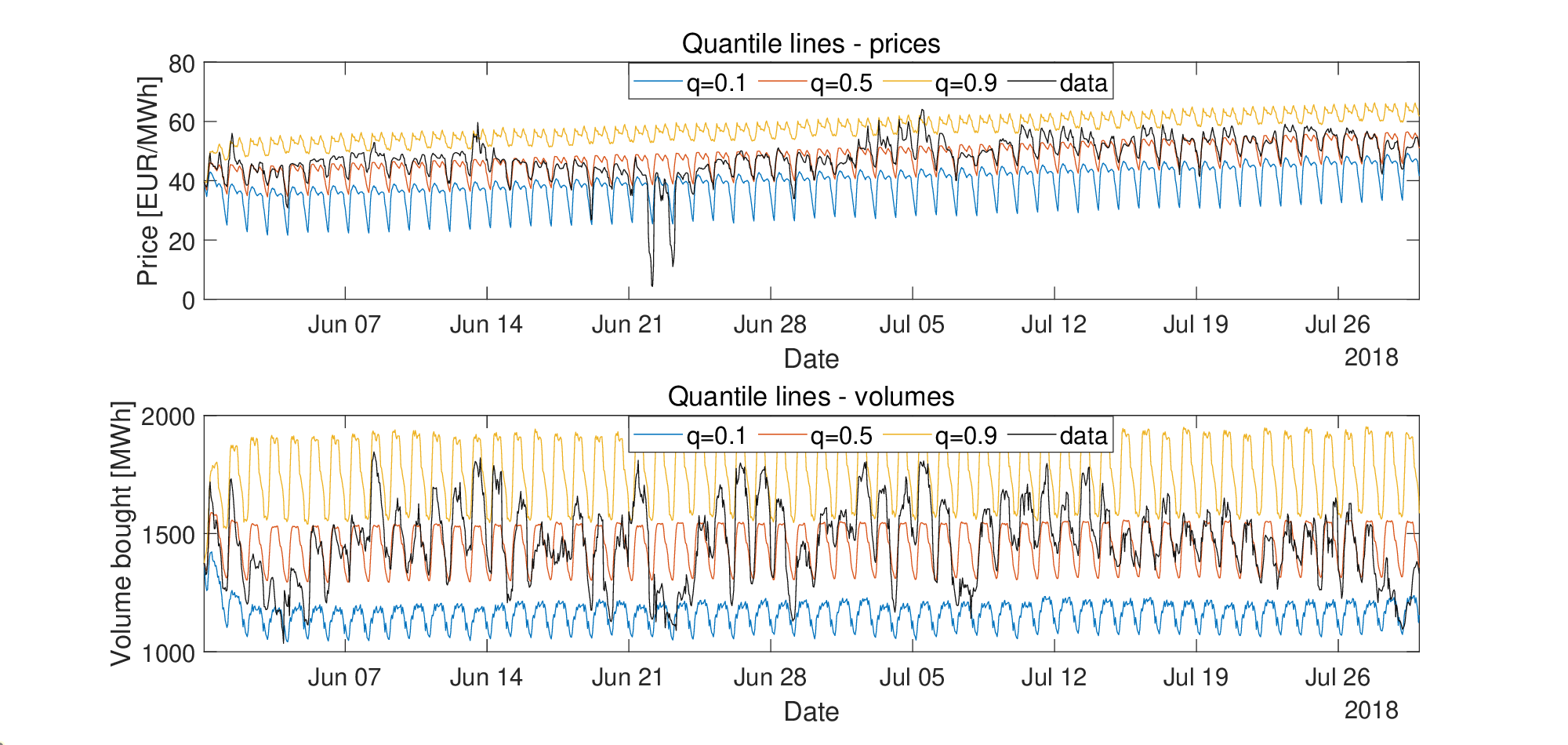}
  \caption{ {Quantile lines, based on the 5000 simulated trajectories of the two-dimensional $\alpha$-stable PAR(1) model with period $T=24$ and with estimated parameters (with added deterministic components) in comparison with the original dataset. }}
  \label{fig:data_qlines}
\end{figure}
\begin{figure}[h!]
    \centering
    \includegraphics[scale=0.37]{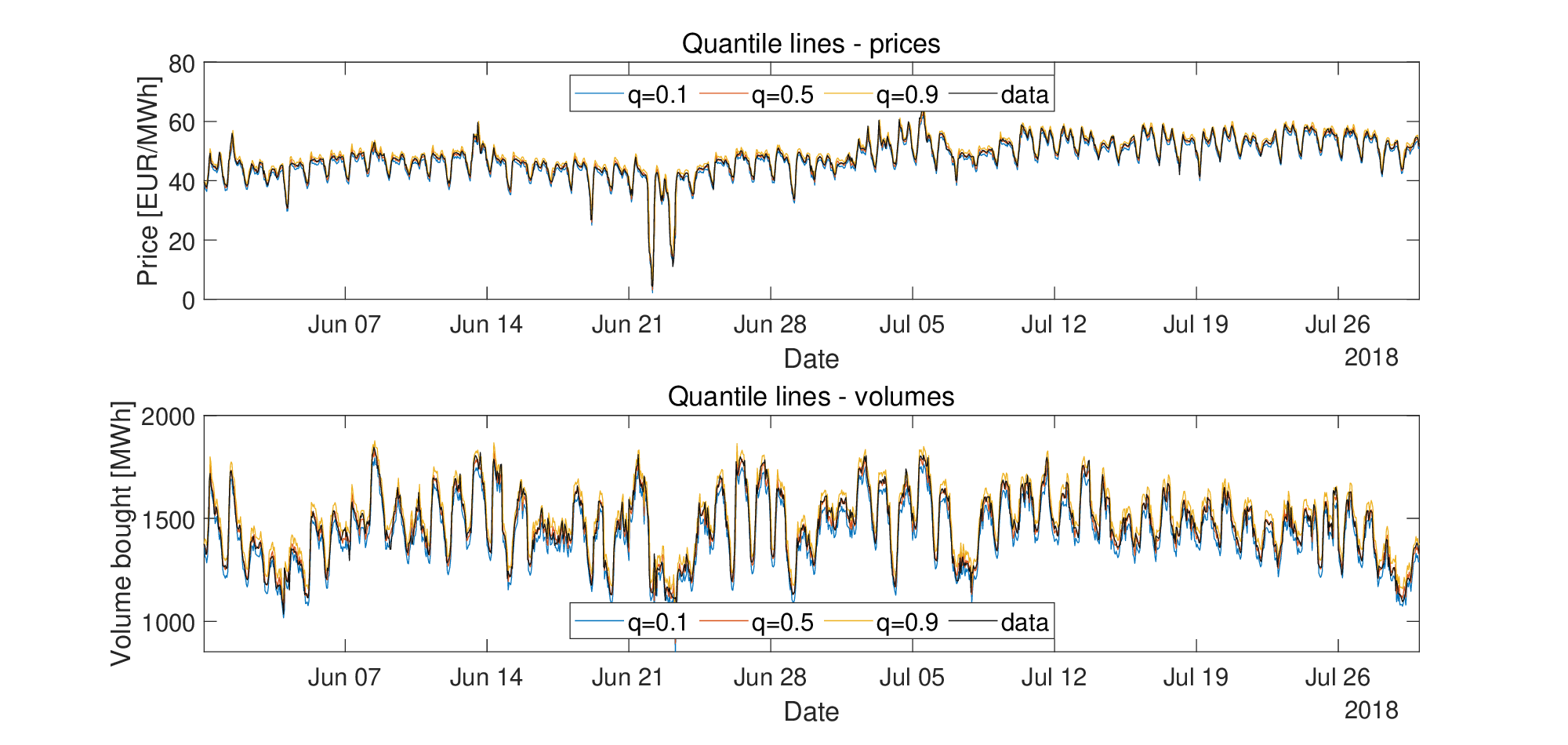}
    \caption{Quantile lines of one-step ahead predictions, based on the 5000 simulated trajectories of the two-dimensional $\alpha$-stable PAR(1) model with period $T=24$ and with estimated parameters (with added deterministic components) in comparison with the original dataset.}
    \label{fig:qlines_1step}
\end{figure}
\section{Conclusions and future study} \label{Conclusions}
In this paper, the multidimensional PAR model with infinite-variance distribution is considered. The special attention we have paid to new estimation methods dedicated to the model. These method are based on the generalized Yule-Walker equations where the classical auto- and cross-covariance functions are replaced by the auto- and cross-covariation functions, properly defined for the $\alpha-$stable models.  The presented simulation results indicate the validity of the proposed methodology. Moreover, we compare the Y-W-CV method with Y-W-T method. The presented simulation study indicates the Y-W-CV technique is more effective than Y-W-T. However, one can conclude that the dispersion of the new estimator is still a challenging task. Finally, the new approach was applied for the real-time series to show the possible applications of the introduced methodology. Although the fitted model is not perfect  for the analyzed data, the presented analysis demonstrates how to proceed in order to fit the multidimensional periodic model with non-Gaussian characteristics.

Although the new approach is introduced and its validity is examined in this paper, there are still a few important open questions. The first one is the proper identification of the model order. In the classical Gaussian case, there are known information criteria that can be useful here. Also in a one-dimensional infinite-variance case, several attempts have been made in this area  \citep{kruczek_physica}. However, according to our knowledge, there is still a space for the new algorithms of the proper model recognition for the multidimensional infinite-variance models.  The second issue is the identification of the period $T$.  In this paper for real data analysis, it is assumed that $T$ is known. However, for many real applications, this value is not given and needs to be estimated. In case when the data are Gaussian or finite-variance-distributed, there are known methods for period identification, like coherent or incoherent statistics based on the autocovariance function \citep{broszkiewicz2004detecting}. In the $\alpha-$stable case, the problem is much more complicated. One of the possible approaches is based on the replacement of the autocovariance function by the adequate dependency measure. The similar problem was discussed in \citep{kruczek2020detect} for vibration data. However this issue needs to be carefully investigated and thus, we decide to consider it as the future study.

The possible extension of the presented methodology is the introduction of the new estimation algorithms for the time-varying AR models when the assumption of Gaussian distribution is not valid. Thus, we see here a huge potential for further research and further application of the introduced methodology. 

\section*{Data Availability Statement}
Not applicable.

\section{Declarations}
\textbf{Funding} \\
The work of P. Giri and  S. Sundar  was supported by Indian Institute of Technology Madras, India under the Project No. SB20210848MAMHRD008558 "Centre for Computational Mathematics and Data Science, Department of Mathematics, IIT Madras Chennai India".
\\The work of  A. Wy{\l}oma{\'n}ska was supported by  the National Center of Science under Opus Grant 2020/37/B/HS4/00120 "Market risk model identification and validation using novel statistical, probabilistic, and machine learning tools".\\
The work of W. {\.Z}u{\l}awi{\'n}ski was supported by National Center of Science under Sheng2 project No.  UMO-2021/40/Q/ST8/00024 "NonGauMech - New methods of processing non-stationary signals (identification, segmentation, extraction, modeling) with non-Gaussian characteristics for the purpose of monitoring complex mechanical structures". \\
\noindent \textbf{Conflicts of interest/Competing interests} \\
The authors have no conflicts of interest to declare that are relevant to the content of this article. \\
\noindent \textbf{Availability of data and material} \\
 {The data were available and downloaded in September 2020.}\\
\noindent \textbf{Code availability} \\
Not applicable\\
\noindent \textbf{Authors' contributions} \\
Conceptualization: [ S. Sundar, A. Wylomanska, P. Giri, A. Grzesiek, W. \.Zu\l awi\'nski], Methodology: [A. Wylomanska,  P. Giri, A. Grzesiek], Formal analysis and investigation: [ P. Giri, A. Grzesiek, W. \.Zu\l awi\'nski],  Writing - original draft preparation: [P. Giri,A. Grzesiek, W. \.Zu\l awi\'nski, A. Wylomanska]; Writing - review and editing: [P. Giri,A. Grzesiek, A. Wylomanska], Supervision: [S. Sundar, A. Wylomanska]


\bibliographystyle{abbrvnat}
\bibliography{bib}
	
\appendix
	{\small \section{The proof of Lemma \ref{covariation_general}}
		 \noindent From \citep{Taqqu}, the cross-covariation given in Eq. (\ref{cross-CV}) can be calculated using the following formula  
			{\begin{equation*}
			\mathrm{CV}(X_r(s),X_l(t))=\frac{1}{\alpha}\frac{\sigma^\alpha(\theta_1,\theta_2)}{\partial\theta_1}\bigg\rvert_{\theta_1=0,\theta_2=1},
			\end{equation*}}
			where $\sigma^\alpha(\theta_1,\theta_2)$ denotes the scale parameter of $\theta_1X_r(s)+\theta_2X_l(t)$. Note that for $s\geq t$\\
			\begin{align*}
			\exp\{-\sigma^\alpha&(\theta_1,\theta_2)\}=\E\left[ \exp\left\{i\left(\theta_1X_r(s)+\theta_2X_l(t)\right)\right\}\right]\\[5pt]&=\exp\Bigg\{-\Bigg(\sum_{j=0}^{-(t-s)+1}\int_{S_m}|\theta_1g_{r1}(s,{s-j+1})s_1+\ldots+\theta_mg_{rm}(s,{s-j+1})s_m|^\alpha\Gamma(ds)\\&\hspace{1cm}+\sum_{j=-(t-s)}^{+\infty}\int_{S_m}|\left(\theta_1g_{r1}(s,{s-j+1})+\theta_2g_{l1}({t},{s-j+1})\right)s_1+\ldots\\&\hspace{1cm}+\big(\theta_1g_{rm}(s,{s-j+1})+\theta_2g_{lm}({t},{s-j+1})\big)s_m|^\alpha\Gamma(ds)\Bigg)\Bigg\}.
			\end{align*}\\
			And consequently 
		\begin{multline*} 
		\mathrm{CV}(X_r(s),X_l(t))=\sum_{j=0}^{+\infty}\int_{S_m}\left(g_{l1}({t},{t-j+1})s_1+g_{l2}({t},{t-j+1})s_2+\ldots+g_{lm}({t},{t-j+1})s_m\right)^{\langle\alpha-1\rangle}\\\left(g_{r1}({s},{t-j+1})s_1+g_{r2}({s},{t-j+1})s_2+\ldots+g_{rm}({s},{t-j+1})s_m\right)\Gamma(ds).
		\end{multline*}
		By proceeding similarly one can show that for $s\leq t$
		\begin{multline*}
		\mathrm{CV}(X_r(s),X_l(t))=\sum_{j=0}^{+\infty}\int_{S_m}\left(g_{l1}({t},{s-j+1})s_1+g_{l2}({t},{s-j+1})s_2+\ldots+g_{lm}({t},{s-j+1})s_m\right)^{\langle\alpha-1\rangle}\\\left(g_{r1}({s},{s-j+1})s_1+g_{r2}({s},{s-j+1})s_2+\ldots+g_{rm}({s},{s-j+1})s_m\right)\Gamma(ds).
		\end{multline*}
	\qed}

{\small \section{An Algorithms to solve system of equations}} \label{algo}
In this section, we present an algorithm (method) to solving the consistent system of equations (\ref{system_est}). Since the normalized-covariation-based matrix is positive semi-definite matrix, therefore it can be singular or non-singular matrix i.e $\widehat{\mathbf{NCV}}^{(v-1)}(0)$ for each $v = 1,\ldots,T$ can be singular or non-singular. If $\hat{NCV}^{(v-1)}(0)$ is singular and system of equations given in Eq. (29) are consistent. Then the solution $\hat{\Theta}(v)$ exists but it is not unique. So in this case, we need a numerical method for solving system of equations. There are many methods available in literature but in this paper we used the latest and very popular numerical method named by Bi-conjugate gradient stabilized method i.e BICGSTAB. This method can be use for non-singular case also.  In both cases, we can use the following method (algorithm) to solve the system of equations (\ref{system_est}). 

\begin{algorithm}[H]
\caption{\textbf{ The algorithm to solving system of equation (\ref{system_est})}}\label{alg:BICGSTAB}
\begin{algorithmic}[H]
\Procedure{$SOLUTION_1$}{$\widehat{\mathbf{NCV}}^{(v-1)}(0)$ ,   $\widehat{\mathbf{NCV}}^v(1)$} \Comment{where $\widehat{\mathbf{NCV}}^{(v-1)}(0)$ and $\widehat{\mathbf{NCV}}^v(1)$ are $m \times m$ matrices}
\State $A_v \gets \widehat{\mathbf{NCV}}^{(v-1)}(0)'$  \Comment{where $\widehat{\mathbf{NCV}}^{(v-1)}(0)'$ is transpose of $\widehat{\mathbf{NCV}}^{(v-1)}(0)$ }
\State $B_v \gets \widehat{\mathbf{NCV}}^v(1)'$   \Comment{where $\widehat{\mathbf{NCV}}^v(1)'$ is transpose of $\widehat{\mathbf{NCV}}^v(1)$ }
\State \textbf{for} $i= 1,2$
\State $\mathbf{\Theta_v}(:,i) \gets BICGSTAB(A_v,B_v(:,i),tol,maxit, M_{1},M_{2})$ \Comment{where $\mathbf{\Theta_v}(:,i)$ is $i^{th}$ column of $\mathbf{\Theta_v}$ and $M = M_1M_2 \sim  A_v$ is preconditioned.}
\State \textbf{end}
\State $\mathbf{\Theta}_v \gets \mathbf{\Theta}_v'$ \Comment{where $\mathbf{\Theta}_v = \mathbf{\Theta}(v)$} 
\State \textbf{return} $\mathbf{\Theta}(v)$
\EndProcedure
\end{algorithmic}
\end{algorithm}
where $BICGSTAB$ is a method (algorithm) with preconditioned to solve linear system of equations $\textbf{A x = b}$ (say), where $\textbf{A}$ is coefficient matrix whether \textbf{A} can be singular or non-singular, $\textbf{x}$ is vector of unknown and $\textbf{b}$ is known vector. It is latest and updated method. For more details of $BICGSTAB$ see \citep{g2,g4,g5,g6}.
{\small \section{Values of estimated parameters for the analyzed data}} \label{ests}

\begin{table}[!h]
	\small\centering
	\caption{ {Estimated values of the fitted model's parameters.}} 
	\vspace{5pt} \centering
\begin{tabular}{||c||c||c||c||c||}
		\hline 
		$v$ & $\widehat{\Theta}_{11}(v)$  & $\widehat{\Theta}_{12}(v)$ & $\widehat{\Theta}_{21}(v)$ & $\widehat{\Theta}_{22}(v)$  \\
		\hline
		\hline
		1 & 1.0226 & 0.0014 & 1.8911 & 0.9637   \\ 
		\hline
		2 & 1.0560 & 0.0005 & 0.6924 & 0.9514   \\ 
		\hline
		3 & 1.0896 & 0.0014 &  -0.1668 & 0.9449   \\ 
		\hline
		4 & 1.0284 & -0.0026 & -0.6131  & 0.9920   \\ 
		\hline
		5 &  0.7265 & 0.0054 &  -1.4507  & 0.9904    \\ 
		\hline
				6 & 0.7508  & 0.0001 & 4.3211  & 0.9230    \\ 
		\hline
				7 & 0.9214 & -0.0002 & 11.4943 &  0.8547   \\ 
		\hline
				8 & 1.0841  & 0.0009 & 6.7091  & 1.0086    \\ 
		\hline
				9 & 0.9445 & -0.0029 & -1.5941  & 0.9819    \\ 
		\hline
				10 & 0.9374 &  -0.0007 & 0.7084 &  0.9002   \\ 
		\hline
				11 & 1.0180  & -0.0007 & -4.2928 & 1.0619   \\ 
		\hline
				12 & 0.8873 & -0.00003 &  -0.7590  & 0.9210    \\ 
		\hline
				13 & 0.9757 & 0.0026 & 3.5708 & 0.9038   \\ 
		\hline
				14 & 0.9979 & -0.0004 & -1.6383  & 1.0344   \\ 
		\hline
				15 & 1.0264 &  -0.0017 & 1.0125 & 0.9504   \\ 
		\hline
				16 & 0.9862 & -0.0019 & -2.6817  &  0.9845   \\ 
		\hline
				17 & 0.9803 & 0.0010 & -3.9416  &  0.9087   \\ 
		\hline
				18 & 1.0461  & 0.0007 &  8.0590 & 0.7278   \\ 
		\hline
				19 & 0.9967 & -0.0004 &  2.2757 &   0.9442   \\ 
		\hline
				20 & 0.9141 & -0.0005 & 1.2990 & 0.8396    \\ 
		\hline
				21 & 0.9555 & -0.00003 & -2.3115 & 0.9922    \\ 
		\hline
				22 & 0.9654 & 0.0005 & 0.7507 & 0.9437    \\ 
		\hline
				23 & 1.0466 & -0.0015 &  0.8990  &   0.9548   \\ 
		\hline
				24 & 1.0265 & 0.0004 &  0.0823  &  0.7418    \\ 
		\hline

		\hline

	\end{tabular}
\label{tab:ests}		
\end{table}
\clearpage

{\small \section{Estimated spectral measure}} \label{spectral_meas}

\begin{figure}[h!]
    \centering
    \includegraphics[width=0.75\textwidth]{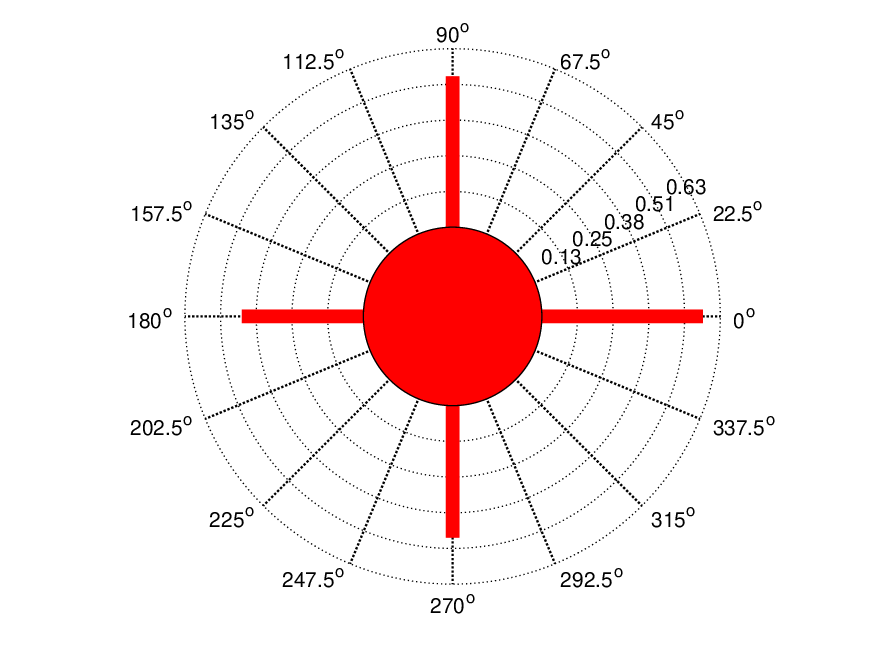}
    \caption{Estimated spectral measure for the residual series (with normalized components) of the fitted model.}
    \label{fig:spect2d}
\end{figure}

\end{document}